\documentclass[secnumarabic,floatfix,nofootinbib,
tightenlines,nobibnotes,aps,twocolumn,prl,10pt]{revtex4-1}

    \usepackage[english]{babel}
   \usepackage[utf8]{inputenc}
    \usepackage[OT1]{fontenc}
  \usepackage{amsmath,amssymb,amsthm,bm}
  \usepackage{appendix}

\usepackage{enumerate}
   \usepackage{makeidx}
  \usepackage{amsfonts}
  \usepackage[pdftex]{graphicx}
\usepackage{float}
	\usepackage{pst-grad} 
	\usepackage{pst-plot} 
 	\usepackage[colorlinks,hyperindex]{hyperref}
    \usepackage{nicefrac}
	\hypersetup      
	{
		colorlinks,%
		citecolor=black,%
		linkcolor=black,%
		urlcolor=black,%
	}

\makeindex

\begin{document}

\newtheorem{Teo}{Theorem}
\renewcommand{\sin}{sin}

\title{Study of  Some Simple Approximations to the  Non-Interacting Kinetic Energy Functional}

\author{Edison X. Salazar}
\email{salazar.xavier@gmail.com}
\affiliation{Universidad Regional Amazónica IKIAM. Km 7 vía Muyuna, Tena,
Ecuador.}

\author{Pedro F. Guarderas}
\affiliation{Universidad Central del Ecuador.  Cdla. Universitaria, Quito, Ecuador}
\author{Eduardo V. Ludeña}
\affiliation{
Center of Nanotechnology Research and Development (CIDNA), ESPOL, Km 30.5 vía Perimetral, Campus
Gustavo Galindo,
Guayaquil, Ecuador.}
\affiliation{Programa Prometeo, SENESCYT, Quito, Ecuador}
\affiliation{ Centro de Química,
Instituto Venezolano de Investigaciones Científicas, IVIC, Apartado 21827, Caracas 1020-A,
Venezuela}
\author{Mauricio Cornejo}
\affiliation{
Center of Nanotechnology Research and Development (CIDNA), ESPOL, Km 30.5 vía Perimetral, Campus
Gustavo Galindo,
Guayaquil, Ecuador.}
\thanks{Permanent Address}
\author{Valentin V. Karasiev}
\affiliation{Quantum Theory Project, Departments of Physics and Chemistry, P.O. Box 118435, University of Florida, Gainesville FL 32611-8435.}

\date{\today}

\begin{abstract}

Within the framework of density functional theory, we present a study of approximations to  the enhancement factor of the
non-interacting kinetic energy functional $T_s[\rho]$.  For this purpose, we employ the model of Liu
and Parr [S. Liu and R.G. Parr, Phys. Rev. A {\bf 55}, 1792 (1997)] based on a series expansion of
$T_s[\rho]$ involving powers of the density. Applications to 34 atoms, at the Hartree-Fock level showed that the
enhancement factors present peaks that are in excellent agreement with those of the exact ones
and  give an accurate description of the shell structure of these atoms. The
application of Z-dependent expansions to represent  some of the terms of these approximation for neutral atoms and  for  positive and negative ions, which allows  $T_s[\rho]$ to be cast in a very simple form,  is also explored.  Indications are given as to how these functionals may be applied to molecules and clusters.\\

  \vspace{0.02cm}

\end{abstract}

\maketitle

%
%

\section{Introduction.}

One of the challenging  problems in density functional theory, DFT, is how to express the
non-interacting  kinetic energy of a quantum mechanical many-body system as a functional of the
density \cite{3.4,1.5,1.10,1.12,1.11,Apro,Iyengar-2001,Karasiev-et-al-2014-Bach,Karasiev-et-al-2015}.
%
%
%
%

Having such a functional is, of course, crucial  for the
implementation of the orbital-free version of  DFT.\cite{Wang-Carter-2000}
%
Let us recall that   in the  Kohn-Sham
\cite{1.3} version of DFT, as a result of writing the non-interacting kinetic energy  as  a
functional of the orbitals there arise  $N$ Kohn-Sham equations whose solution becomes progressively
more difficult as  $N$ gets larger. For this reason, a treatment that dispenses with  orbitals and
is based on the use of a kinetic energy functional which only depends on  the density has been
proposed as an alternate way for handling this problem. To  focus on an orbital-free functional for
the energy would  certainly lower the  computational cost and would permit to extend   the domain of
application of DFT to large many-particle systems as
all one has to solve is a single equation for the density, regardless of the value of $N$.

To find an adequate  density functional for the kinetic energy is a difficult task because, due to
the virial theorem,  the kinetic energy is  equal in magnitude to the total energy.  Hence this
functional must have the same level of accuracy as that of the total energy (in contrast with  the
exchange and correlation functionals in DFT which only comprise small portions of the total energy).
This fact probably explains why since the first works of Thomas \cite{1.1} and Fermi \cite{1.2} in
1927  and in spite of the continued   efforts carried out  throughout many decades, (for reviews see
Refs.\cite{Wang-Carter-2000,2.1,Apro}) still no definite and  satisfactory approximation to this
functional  has been found.  The search for suitable approximations is, however, an ongoing
activity.\cite{Lindmaa-2014,Higuchi-2014,2014-Laricchia,Xia-carter-2015,sergeev,DellaSala-2015,
Espinosa-Leal-2015,Finzel-2015,Li-et-al-2015} Although an exact analytical expression for the non-interactive kinetic energy
functional is still lacking, the exact form of this functional, however, is known and  can be
derived  from
general principles (see, for example, the derivation given in the context of the local-scaling
transformation version of DFT
\cite{EVL-CanJChem-1996,1.9,1.4})
%
%
%
%
%
This exact form corresponds to:
\begin{equation}
T_s[\rho]= \frac{1}{8}\int d\vec r\frac{|\nabla\rho(\vec r)|^2}{\rho(\vec r)} + \frac{1}{2}\int
d\vec r\rho^{\nicefrac{5}{3}}(\vec r)A_N[\rho(\vec r);\vec r], \label{eq:0.1}
\end{equation}
where the first term is the Weizsäcker term \cite{2.2} and where the second contains the
Thomas-Fermi function $\rho^{\nicefrac{5}{3}}(\vec r)$ times the enhancement factor
$A_N[\rho(\vec r); \vec r]$ where  $\rho(\vec r)$ is
the one-electron density of the system. Clearly, as is seen in the above expression, the challenge in modelling  $ T_s[\rho]$ is shifted to that of attaining adequate approximations for the enhancement factor $A_N[\rho(\vec r);\vec r]$, which is considered to be expressible as  a functional of $\rho$.

Among the alternatives produced over the years to represent the  non-interactive kinetic energy functional \cite{Apro}, we focus in this paper on the one introduced by
Liu and Parr \cite{3.4},  which is given as a power series  of the
density $\rho(\vec r)$. This series  generates an explicit  expression for the  enhancement factor as a functional  of the
one-particle density. A variational calculation based on this expansion has recently been given by Kristyan.\cite{Kristian-2013}
%

 In this paper we analyze the representation of the
enhancement factor  given by the Liu-Parr series expansion and compare it with  the exact values extracted from an orbital expression.  This is done  for the atoms of the first, second and third row of the periodic table.  In addition, we explore the possibility of simplifying the Liu-Parr functional by introducing  $Z$-dependent expressions for some of the integrals containing $\rho(\vec r)$. Finally, bearing in mind that  the mathematical framework is presented rather concisely in the
original Liu and Parr's paper, \cite{3.4} we  include in Appendix A of the present work a more
extended demonstration of their second theorem. We expect that this may contribute to a better
understanding of the Liu and Parr approach and foster its applications.

\section{The  Enhancement Factor.}

\subsection{Some properties of the enhancement factor}
From an information theory perspective,\cite{Sears-Parr-Dinur}
%
%
the Weizsacker term in Eq. (\ref{eq:0.1}) is  local. Clearly then, since  the non-local part  of the kinetic energy functional must be embodied in the non-Weizsacker term, this non-locality must be ascribed to the  enhancement factor. In fact, as was pointed out by Ludeña,\cite{Ludena-1982}
%
the non-Weizsacker term  contains the derivative of the correlation factor for the Fermi hole (see
Eq.[38] of Ref.\cite{Ludena-1982}). Hence, the enhancement factor  contains terms responsible for
localizing electrons with the same spin in different regions of space giving rise to shell
structure.  This  phenomenon stems from the non-locality of the Fermi hole which may be described in
terms of charge depletions followed by charge accumulations producing  polarizations at different
distances.\cite{Ludena-et-al-JCP-2004}
%
%
This non-locality of the kinetic energy functional is well represented by orbital expansions
\begin{eqnarray}
T_s\left[\{\phi_i\}^{N}_{i=1}\right] &=& \frac{1}{2}\sum^{N}_{i=1}\int d\vec r\ \nabla\phi^*_i(\vec r)\nabla\phi_i(\vec r)\label{eq2}
\\
&=&
-\frac{1}{2}\sum^{N}_{i=1}\int d\vec r\ \phi^*_i(\vec r)\nabla^2\phi_i(\vec r) \label{eq3} \\
& & + \frac{1}{4}\int d\vec r\ \nabla^2\rho (\vec r)
\nonumber
\end{eqnarray}
given in terms of  gradients in Eq. (\ref{eq2}) or of  Laplacians in Eq. (\ref{eq3}).
Combining  Eqs. (1) and the gradient representation of Eq.(2), we obtain the following exact  orbital representation for the enhancement factor:
\begin{eqnarray}
 A_N\left[\rho(\vec r),\left\{\phi_i\right\};\vec r\right]&=&
 {{2}\over{\rho^{5/3}({\vec r})}}
 \Big({{1}\over{2}}\sum^{N}_{i=1}\nabla\phi^{*}_i(\vec r)\nabla\phi_i(\vec r)
 \nonumber\\
& &-{{1}\over{8}}{{\vert \nabla\rho(\vec r)\vert^2}\over{\rho(\vec r)}}
\Big)
\label{eq:0.6}
\end{eqnarray}
Obviously, when modelling the enhancement factor in terms of a series  of the one-particle density one would like to reproduce  the same characteristics it manifests in an orbital representation.  Thus, in addition to  yielding a desired accuracy for the calculated values of the non-interacting kinetic energy, the approximate enhancement factor should  satisfy the positivity condition\cite{Tal-Bader}:
\begin{equation}
 A_N\left[\rho(\vec r);\vec r\right]\geq 0  \qquad {\rm for }\quad {\rm all} \qquad  \vec r 
 \label{pos}
\end{equation}
and should also be capable of generating   shell structure.  In this respect, let us note that 
$A_N\left[\rho(\vec r),\left\{\phi_i\right\};\vec r\right]$ as given  by Eq. (\ref{eq:0.6})  differs
 by just  a constant from the function $\chi({\vec r})= D({\vec r})/D_h({\vec r})$ introduced in 
the definition of the electron localization function, ELF (see Ref. \cite{Becke-Edgecombe-1990}).
%
%
Moreover, bearing in mind that ELF and similar functions have been successfully related to the appearance of shell structure in  atoms and molecules\cite{Savin-1997,Savin-2005-ELF,Savin-2005-Theochem,Gatti-2005,Navarrete-Lopez-2008,Contreras-Garcia-2011,Rincon-et-al-2011,DeSilva-2013,Causa-2013}
it is clear that any proposed model of the enhancement factor  must also comply with this requirement.




A popular generalized gradient approximation (GGA)  for the kinetic energy takes the following form \cite{Perspectives}
%
%
\begin{equation}
  T_s[\rho] = T_W[\rho] + \int d\vec r \rho^{\nicefrac{5}{3}}(\vec r)F[s(\vec r)],\label{eq:3.4}
\end{equation}
where $T_W[\rho]$ is the Weiszäcker term \cite{2.2}, and the second term is the Pauli
term \cite{2.1} containing the GGA enhancement factor $F[s(\vec r)]$ which depends on $ s(\vec r) $
\begin{equation}
  s(\vec r) = \frac{|\nabla\rho(\vec r)|}{2(3\pi^2)^{\nicefrac{1}{3}}\rho^{\nicefrac{4}{3}}(\vec r)},
\end{equation}
The variable $s(\vec r)$, the reduced density gradient,  describes the rate of variation of the electronic density, i.e, large values of $s(\vec r)$ correspond to fast variations on the electron density and vice versa \cite{C3-GEA}. The above approximation to the Pauli term containing GGA $F[s(\vec r)]$ factors are at the basis of the conjoint gradient expansion to the kinetic energy
introduced by Lee, Lee and Parr \cite{Lee-Lee-Parr}.
%

A full review of the functionals of the kinetic energy expressed in terms of the density and its derivatives is given by
Wesolowski \cite{Apro}.  For some more recent representations of the enhancement factor of the
non-interacting kinetic energy as a functional of $\rho$ and its derivatives $\nabla\rho$,
$\nabla^2\rho$, etc., see Refs.
\cite{Lapla1,par-z,Karasiev-et-al-2009,2014-Laricchia,sergeev,DellaSala-2015,
Espinosa-Leal-2015}

In the present work we examine, however, a different approximation to the enhancement factor, more
in line with Eq. (\ref{eq:0.1}), namely, a representation of $A_N[\rho(\vec r);\vec r]$ as a local
functional of the density.\cite{3.4}


\subsection{An approximate representation of the enhancement factor}

We adopt the Liu and Parr \cite{3.4}  expansion of the non-interacting kinetic energy functional given
in terms of homogeneous functionals of the one-particle density:
\begin{equation}
  T_{LP97}[\rho]=\sum^n_{j=1}C_j\left[\int d\vec r\rho^{[1+(2/3j)]}(\vec r)\right]^j,\label{C4:18}
\end{equation}
where in $T_{LP97}[\rho]$ the subindex $LP97$ stands for Liu and Parr and the year of publication, 1997.
Following  the original work, we truncate  Eq. (\ref{C4:18}) after $j=3$ and obtain:
\begin{align}
  T_{LP97}[\rho]=&\ C_{T_1}\int d\vec r \rho^{\nicefrac{5}{3}}(\vec r)\nonumber\\
   &\ + C_{T_2}\left[\int d\vec r \rho^{\nicefrac{4}{3}}(\vec r)\right]^2\nonumber\\
   &\ +  C_{T_3}\left[\int d\vec r \rho^{\nicefrac{11}{9}}(\vec r)\right]^3 .\label{eq:0.3}
\end{align}
Liu and Parr\cite{Apro} determine the coefficients $C_{T_j}$'s by least-square fitting setting 
$\rho=\rho_{HF}$, the Hartree-Fock density and obtain:  $C_{T_1} = 3.26422$, $C_{T_2} = -0.02631$
and  $C_{T_3} = 0.000498$  (a typographical error in the values of the coefficients in the original
paper, has been corrected). Clearly, this expansion provides a very simple way to express  the
kinetic energy as a local functional of the density.
%


In this context, an approximate expression for $A_N[\rho(\vec r); \vec r]$ as a functional of the one-particle density can be found from  Eqs. (\ref{eq:0.3}) and (\ref{eq:0.1}):
\begin{align}
  A_{N,Appr}[\rho] =&\ 2\left( \vphantom{\frac{1}{8}\frac{|\nabla\rho(\vec r)|^2}{\rho^{\nicefrac{8}{3}}(\vec r)}}  C_{T_1} + C_{T_2}\rho^{-\nicefrac{1}{3}}(\vec r)\int d\vec r \rho^{\nicefrac{4}{3}}(\vec r)\right.\nonumber\\
    &\ +  C_{T_3}\rho^{-\nicefrac{4}{9}}(\vec r)\left[\int d\vec r \rho^{\nicefrac{11}{9}}(\vec r)\right]^2\nonumber\\
    &\ -\left.\frac{1}{8}\frac{|\nabla\rho(\vec r)|^2}{\rho^{\nicefrac{8}{3}}(\vec r)}\right)\label{eq:0.7}
\end{align}

\subsection{Local corrections to the enhancement factor.}

From equations (\ref{eq:0.6}) and (\ref{eq:0.7}), we obtain the graphs for the exact and
approximate enhancement factors, respectively. These graphs are displayed in  Figures 1 through
6 for the  Na, Al, Ar, Fe, Ni and Kr atoms. One can see that the exact enhancement factor is a positive function, in contrast to the approximate one  which shows negative regions  in  violation of the
positivity condition which must be met by  $A_N[\rho(\vec r);\vec r]$.\cite{3.5}

Due to the fact that the kinetic energy is not uniquely defined  and, as illustrated by Eqs.
(\ref{eq2}) and (\ref{eq3}), that  there are expressions  that yield  locally different kinetic
energy densities but which  integrate to the same value, it is possible to modify the non-positive
approximate  enhancement factor by adding a term that does  not alter the expected value of the
non-interacting kinetic energy but which  locally contributes to make the enhancement factor
positive. This is an acceptable procedure in view  of the non-uniqueness in the  definition of the
local kinetic energy expressions.\cite{C3-LKEIQM,C3-RLKE,C3-HAITLKE}
%
%
%

In this vein, we  add to non-interacting  kinetic energy expression  a term given by the Laplacian of the density times an arbitrary constant $\lambda$, where $\lambda$ is a real number:
\begin{align}
  T_s[\rho]=&\ C_{T_1}\int d\vec r \rho^{\nicefrac{5}{3}}(\vec r) + C_{T_2}\left[\int d\vec r \rho^{\nicefrac{4}{3}}(\vec r)\right]^2 \nonumber\\
  &+  C_{T_3}\left[\int d\vec r \rho^{\nicefrac{11}{9}}(\vec r)\right]^3 +\lambda\int d\vec r\ \nabla^2\rho(\vec r).\label{eq:0.8}
\end{align}
This leads to the following new expression for the approximate enhancement factor:
\begin{align}
  A_{N}[\rho] =&\ 2\left( \vphantom{\lambda\frac{\nabla^2\rho}{\rho^{\nicefrac{5}{3}}}} C_{T_1} + C_{T_2}\rho^{-\nicefrac{1}{3}}(\vec r)\int d\vec r \rho^{\nicefrac{4}{3}}(\vec r)\right.\nonumber\\
    &+  C_{T_3}\rho^{-\nicefrac{4}{9}}(\vec r)\left[\int d\vec r \rho^{\nicefrac{11}{9}}(\vec r)\right]^2\nonumber\\
    &+\left.\lambda\frac{\nabla^2\rho(\vec r)}{\rho(\vec r)^{\nicefrac{5}{3}}}-\frac{1}{8}\frac{|\nabla\rho(\vec r)|^2}{\rho^{\nicefrac{8}{3}}(\vec r)}\right).\label{eq:0.9}
\end{align}
With the addition of this term we see that the kinetic energy value does not alter, because the integral of the Laplacian of the density is zero.\cite{Green}
%
%
The improvements that this added term brings about on  the approximate enhancement factor for the  Na, Al, Ar, Fe, Ni and Kr atoms are shown in Figs. [\ref{fig:FactorNaApr1}] through [\ref{fig:FactorKrApr3}].\\
\begin{figure}[!htbp]
\centering\begin{tabular}{|p{8.5cm}|}
            \hline
            \includegraphics[width = 8.5cm, height = 5.5cm]{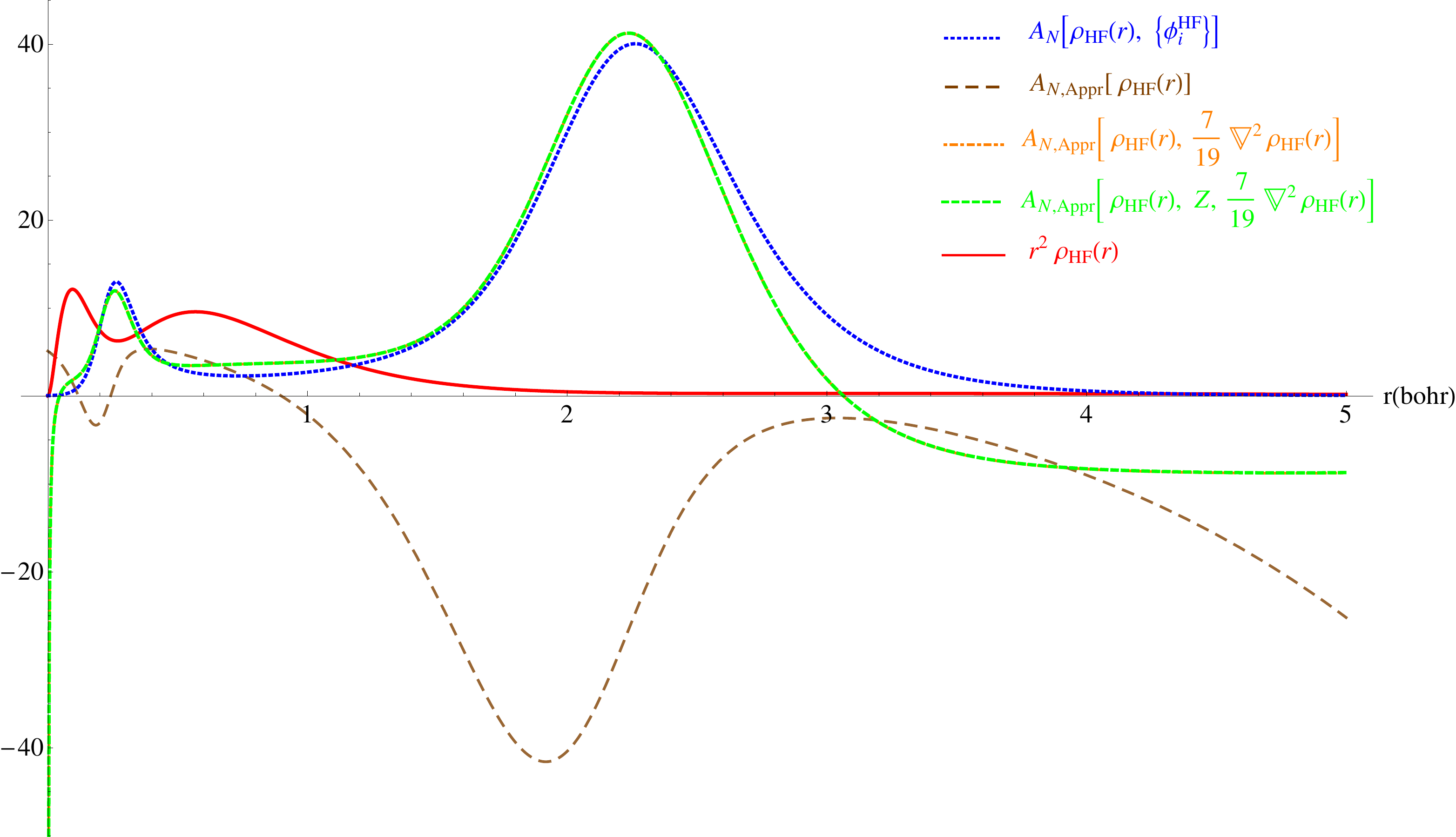}\\
            \hline
          \end{tabular}
\caption{ ``Exact'' enhancement factor (dotted blue), approximate enhancement factor (dashed brown),
approximate enhancement
factor with $\lambda = \frac{7}{19}$ (dotted green), approximate enhancement factor with 9th-degree
$Z$ polynomial and $\lambda = \frac{7}{19}$ (dotted orange), and  radial distribution function of
the
density (full red) for the Na atom}
\label{fig:FactorNaApr1}
\end{figure}
%
brown),
9th-degree
%
\begin{figure}[!htbp]
\centering\begin{tabular}{|p{8.5cm}|}
            \hline
             \includegraphics[width = 8.5cm, height = 5.5cm]{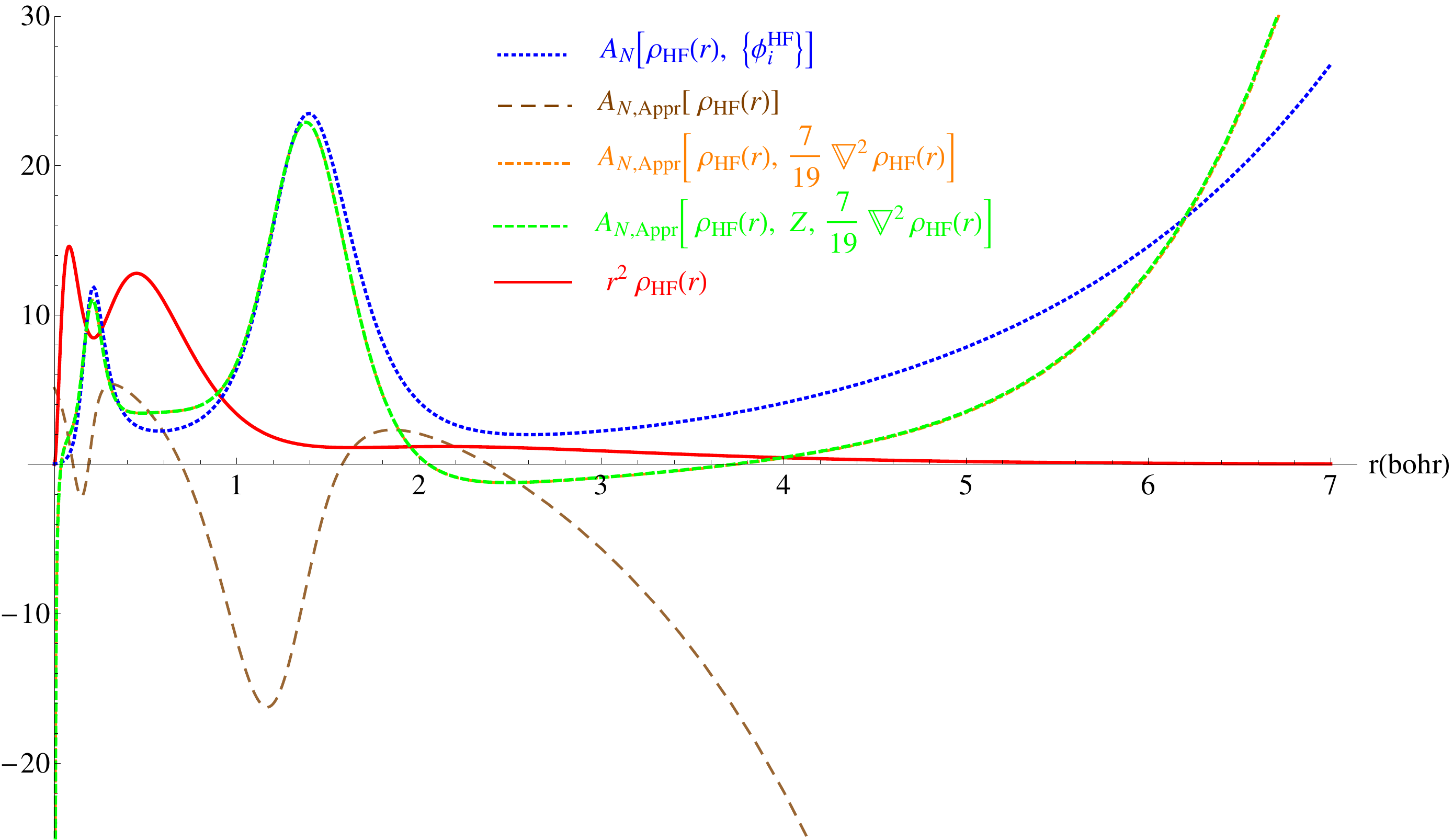}\\
            \hline
          \end{tabular}
\caption{``Exact'' enhancement factor (dotted blue), approximate enhancement factor (dashed brown),
approximate enhancement
factor with $\lambda = \frac{7}{19}$ (dotted green), approximate enhancement factor with 9th-degree
$Z$ polynomial and $\lambda = \frac{7}{19}$ (dotted orange), and 
radial distribution function of the
density (full red) for the  Al atom}
\label{fig:FactorAlApr2}
\end{figure}
%
brown),
9th-degree
%
\begin{figure}[!htbp]
\centering\begin{tabular}{|p{8.5cm}|}
            \hline
            \includegraphics[width = 8.5cm, height = 5.5cm]{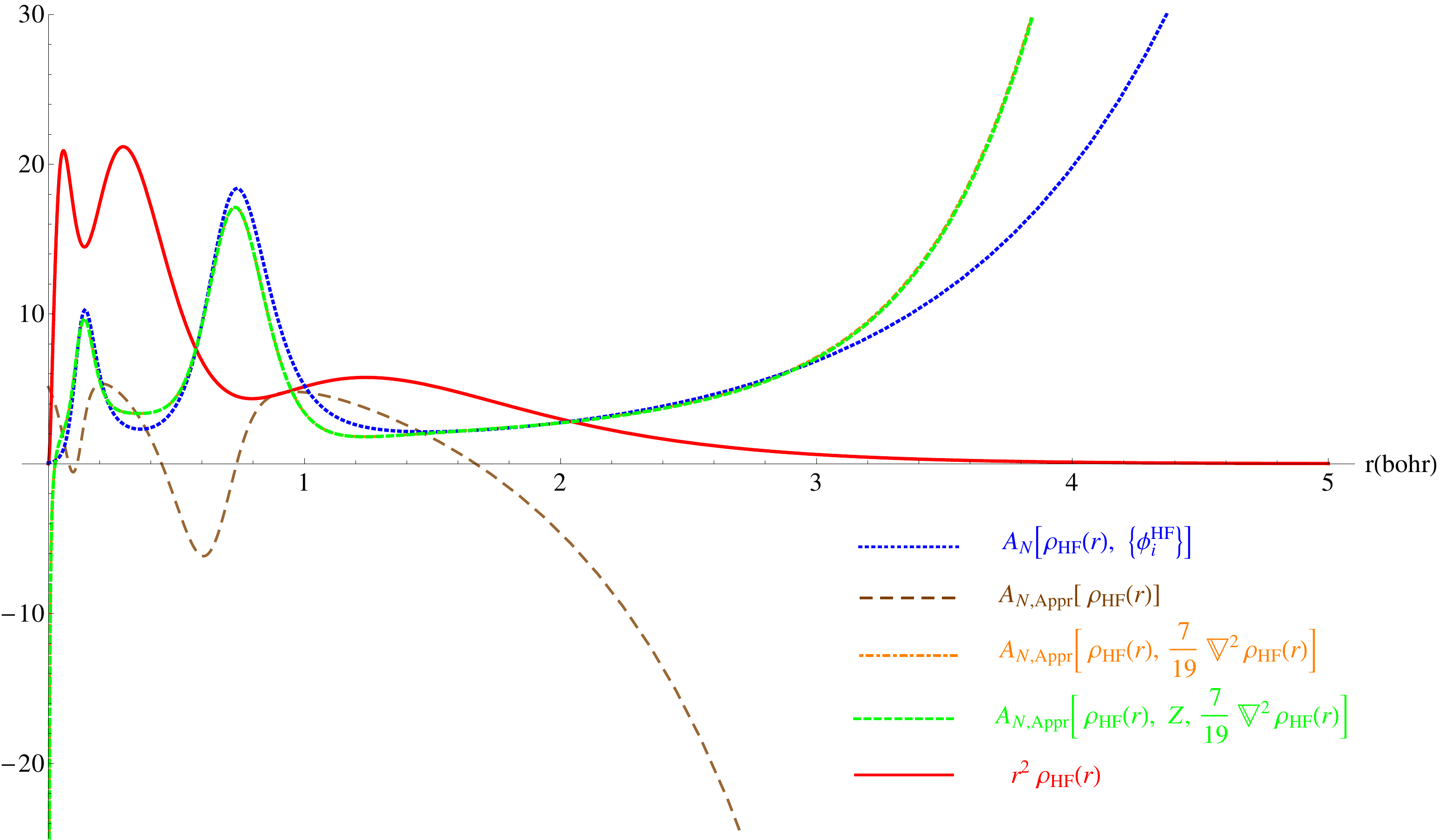}\\
            \hline
          \end{tabular}
\caption{``Exact'' enhancement factor (dotted blue), approximate enhancement factor (dashed brown),
approximate enhancement
factor with $\lambda = \frac{7}{19}$ (dotted green), approximate enhancement factor with 9th-degree
$Z$ polynomial and $\lambda = \frac{7}{19}$ (dotted orange), and  radial distribution function of
the
density (full red)for the Ar  atom}
\label{fig:FactorArApr3}
\end{figure}
%
brown),
9th-degree
%
\begin{figure}[!htbp]
\centering\begin{tabular}{|p{8.5cm}|}
            \hline
            \includegraphics[width = 8.5cm, height = 5.5cm]{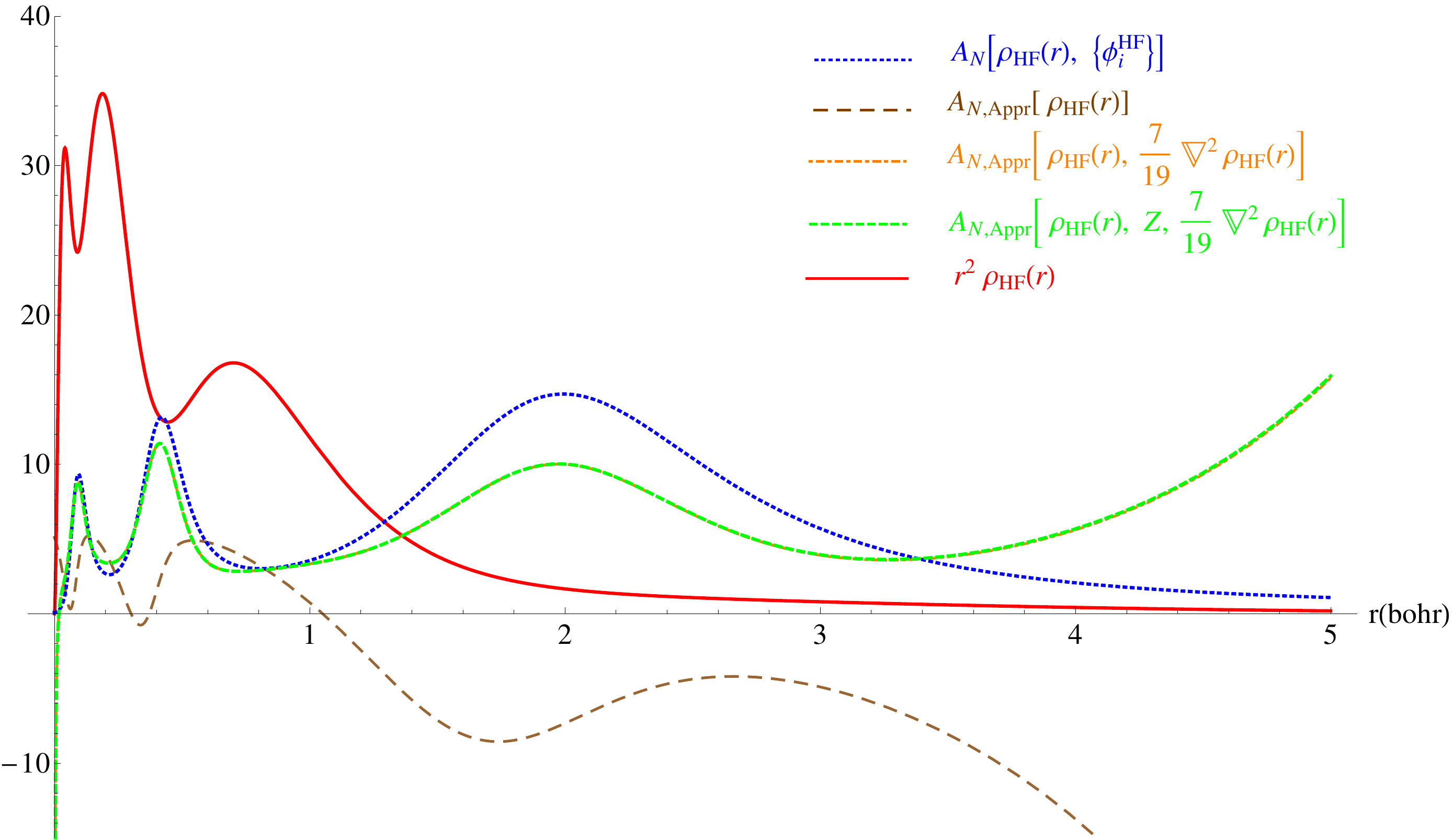}\\
            \hline
          \end{tabular}
\caption{``Exact'' enhancement factor (dotted blue), approximate enhancement factor (dashed brown),
approximate enhancement
factor with $\lambda = \frac{7}{19}$ (dotted green), approximate enhancement factor with 9th-degree
$Z$ polynomial and $\lambda = \frac{7}{19}$ (dotted orange), and  radial distribution function of
the
density (full red) for the Fe  atom}
\label{fig:FactorFeApr3}
\end{figure}
%
brown),
9th-degree
%
\begin{figure}[!htbp]
\centering\begin{tabular}{|p{8.5cm}|}
            \hline
            \includegraphics[width = 8.5cm, height = 5.5cm]{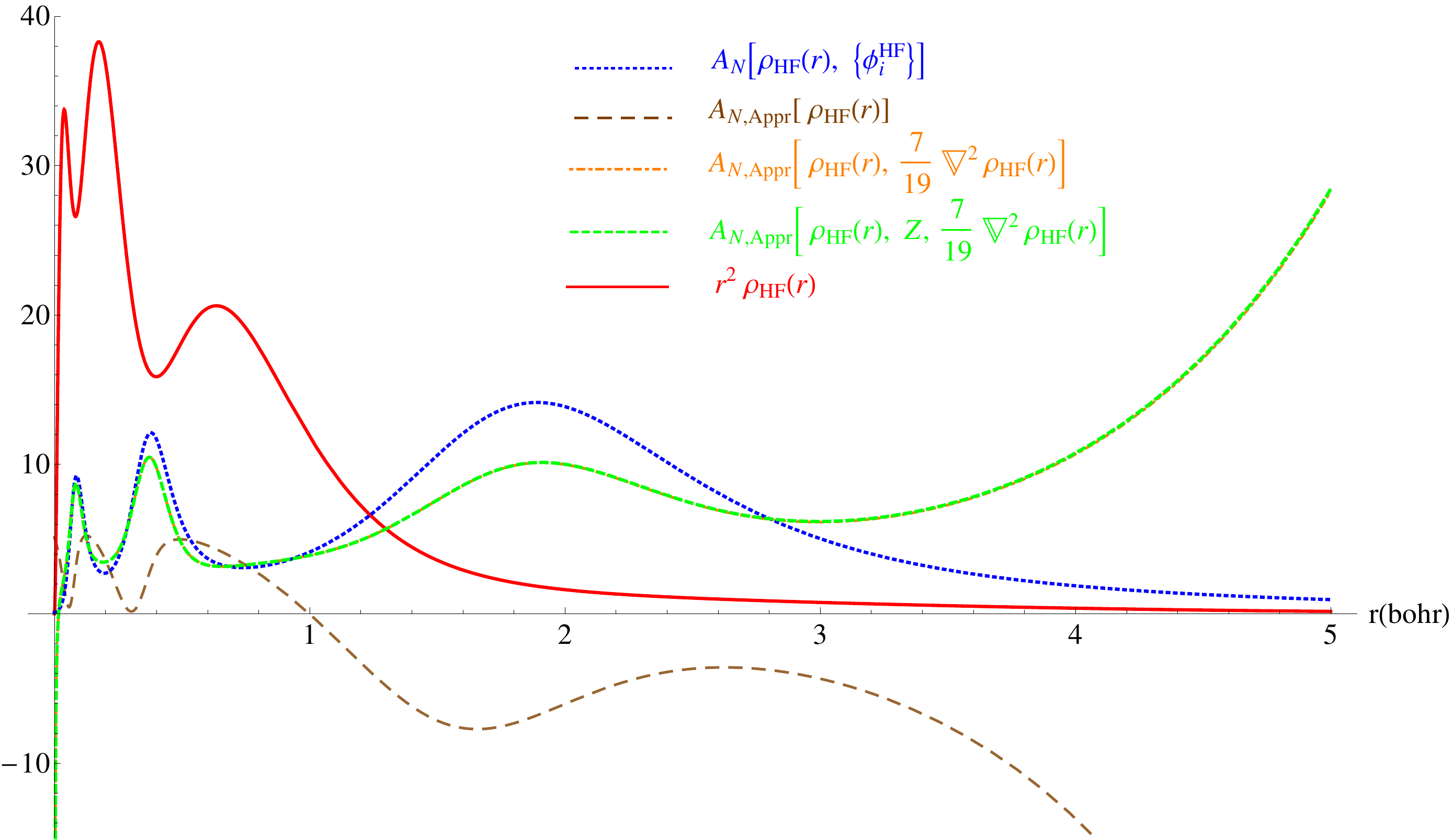}\\
            \hline
          \end{tabular}
\caption{``Exact'' enhancement factor (dotted blue), approximate enhancement factor (dashed brown),
approximate enhancement
factor with $\lambda = \frac{7}{19}$ (dotted green), approximate enhancement factor with 9th-degree
$Z$ polynomial and $\lambda = \frac{7}{19}$ (dotted orange), and  radial distribution function of
the
density (full red) for the Ni atom}
\label{fig:FactorNiApr3}
\end{figure}
%
brown),
9th-degree
%
\begin{figure}[!htbp]
\centering\begin{tabular}{|p{8.5cm}|}
            \hline
            \includegraphics[width = 8.5cm, height = 5.5cm]{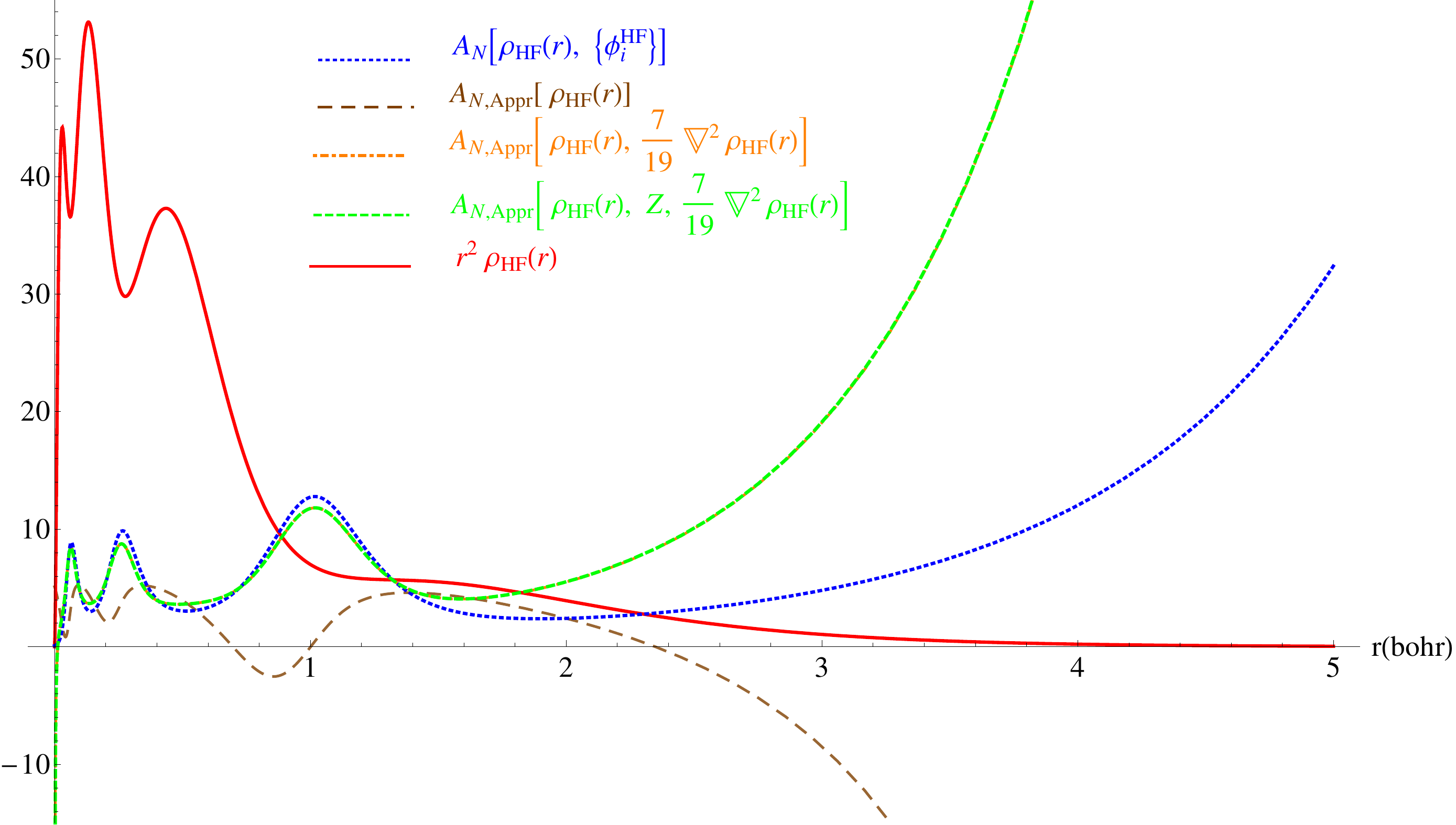}\\
            \hline
          \end{tabular}
\caption{``Exact'' enhancement factor (dotted blue), approximate enhancement factor (dashed brown),
approximate enhancement
factor with $\lambda = \frac{7}{19}$ (dotted green), approximate enhancement factor with 9th-degree
$Z$ polynomial and $\lambda = \frac{7}{19}$ (dotted orange), and  radial distribution function of
the
density (full red) for the Kr atom}
\label{fig:FactorKrApr3}
\end{figure}
%

We note that the new $\lambda$-dependent  enhancement factors closely reproduce the behavior of the exact ones in the regions  where the highest peaks are located.   In all cases the agreement is extremely good both for the first and second shells.  For Ni and Fe, the approximate enhancement factors are slightly below the exact ones in the region corresponding to the third shell.  For Kr, however, the agreement is quite good for all shells.

Concerning the asymptotic behavior of the  enhancement factors, we observe that in the region where $r \rightarrow 0$ the approximate  $\lambda$-dependent enhancement factors  become rapidly negative in all cases studied. It is expected, however, that   
this divergence in $A_N$ will not necessarily  produce a problem for interatomic forces as these 
forces in molecular dynamics are calculated by using the orbital-free analogue of
the Hellmann-Feynman theorem (see Eqs. (21)-(23) in Ref. \cite{ProfessAtQE})
According to these equations the forces are defined by the density and the external potential and thus, they  do not depend on the kinetic energy density local behavior.
On the other hand, in the  region where $r$  becomes large, i.e., outside the atomic shells, the behavior of the approximate factor follows the trend of the exact ones for the cases of Na, Al, Ar and Kr, although in the latter case, the approximate factor grows more pronouncedly that the exact one.  In the case of Ni and Fe, however, we observe a divergence in the behavior of the tails of the approximate enhancement factors. Let us mention, however, that divergences in the tail region are not relevant and do not contribute to the kinetic energy value due to the fact that  these divergences are suppressed by the
exponentially decaying density tail.


\section{Approximation to the enhancement factor through $Z$-dependent polynomials}

We see that the enhancement factor $A_{N,Appr}$, Eq. (\ref{eq:0.7}), depends on  two integrals, $\int
d\vec r\ \rho^{\nicefrac{4}{3}}(\vec r)$ and $\int d\vec r\ \rho^{\nicefrac{11}{9}}(\vec r)$. 
The values of these integrals, evaluated with $\rho=\rho_{HF}$,  are  functions of the
atomic number $Z$.  We have selected to display this $Z$-dependent behavior in Figs. [\ref{fig:43p}] and [\ref{fig:119n}].
In these figures, the dots lying on the blue lines represent the values of the $4/3$ and $11/9$ integrals for the neutral atoms, respectively.  These values are interpolated using    the polynomial expansions:
$\int d\vec r\ \rho^{\nicefrac{4}{3}}(\vec r)\approx P_{\nicefrac{4}{3}}(Z^n)$ and
$\int d\vec r\ \rho^{\nicefrac{11}{9}}(\vec r)\approx P_{\nicefrac{11}{9}}(Z^n)$ where $n$ is the 
degree of the $Z$ polynomial. The blue lines  in frames (a) of Figs. [\ref{fig:43p}]  and [\ref{fig:119n}] represents the  approximations given by the third-degree polynomial $Z^3$. Similarly, the blue lines in frames (b) of these figures correspond to the  $Z^9$ polynomial approximation. 

    \begin{figure*}[!htbp]
        \begin{minipage}[l]{1.0\columnwidth}
            \centering\begin{tabular}{|p{8.5cm}|}
            \hline
             \includegraphics[width = 8.5cm, height = 5.5cm]{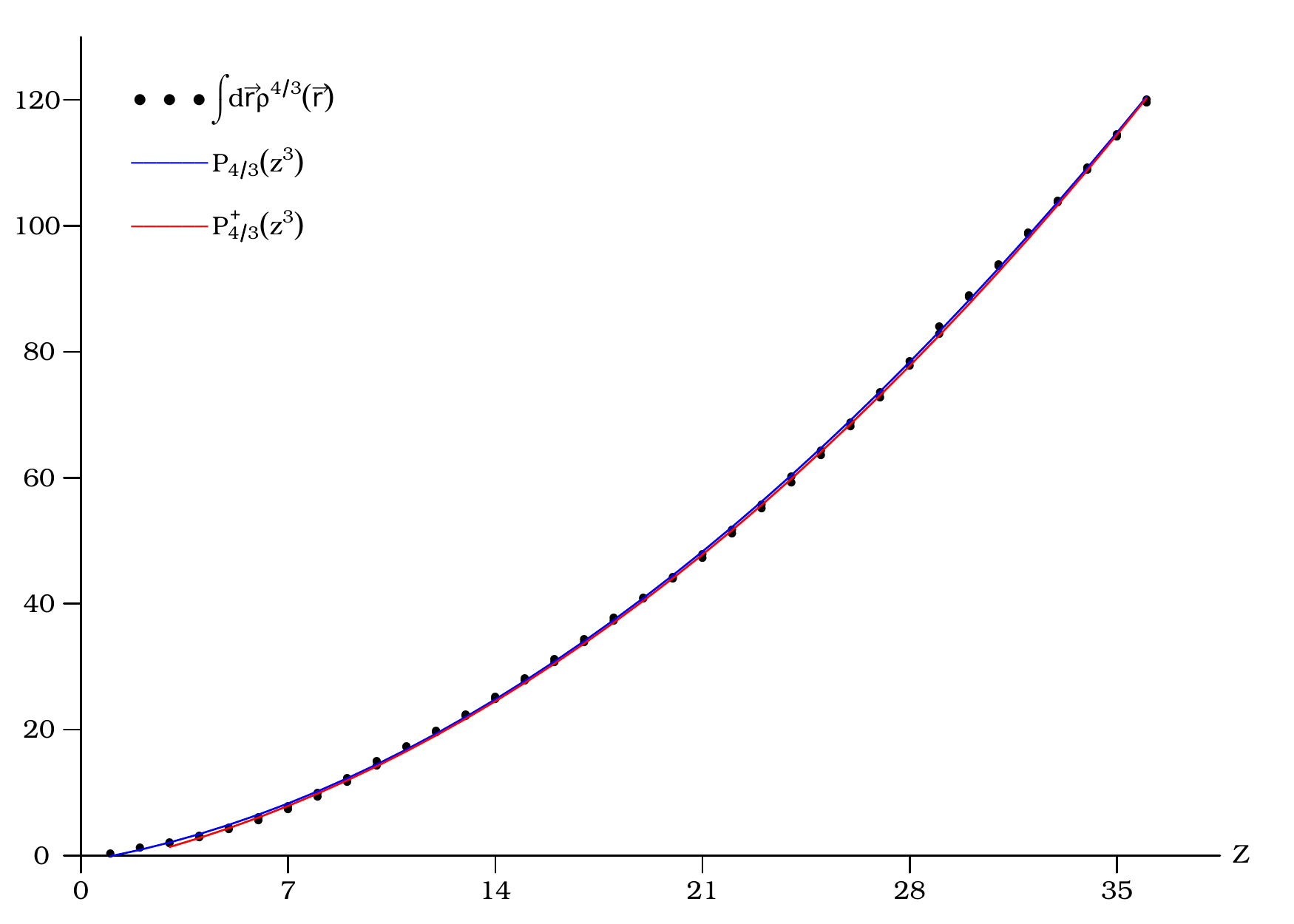}\\
             \hline
          \end{tabular}\\
          \vspace{0.5cm} (a)
        \end{minipage}
        \hfill{}
        \begin{minipage}[r]{1.0\columnwidth}
           \centering\begin{tabular}{|p{8.5cm}|}
            \hline
             \includegraphics[width = 8.5cm, height = 5.5cm]{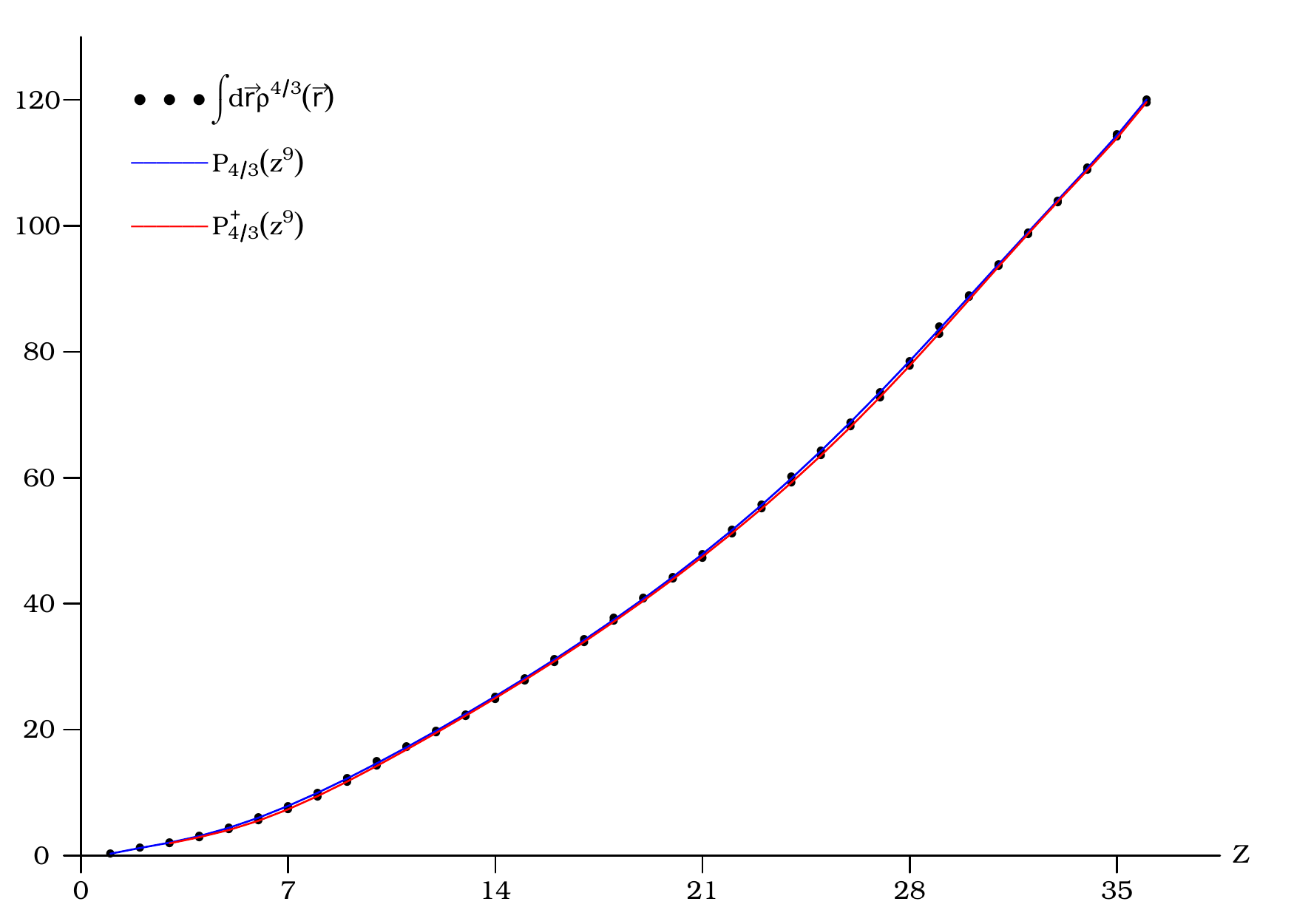}\\
             \hline
          \end{tabular}\\
          \vspace{0.5cm} (b)
        \end{minipage}
        \caption{Interpolation curves for the values of the $\int d\vec r\
\rho^{\nicefrac{4}{3}}(\vec r)$ through: (a) a 3rd degree polynomial  $P_{\nicefrac{4}{3}}(Z^3)$
(full blue) for 36 atoms and a 3rd degree polynomial  $P^+_{\nicefrac{4}{3}}(Z^3)$ (full red) for
34 positive ions, and (b) a 9th degree polynomial $P_{\nicefrac{4}{3}}(Z^9)$ (full blue) for 36
atoms and a 9th degree polynomial $P^+_{\nicefrac{4}{3}}(Z^9)$ (full red) for 34 positive
ions.}\label{fig:43p}
    \end{figure*}
     \begin{figure*}[!htbp]
        \begin{minipage}[l]{1.0\columnwidth}
      \centering\begin{tabular}{|p{8.5cm}|}
	      \hline
		\includegraphics[width = 8.5cm, height = 5.5cm]{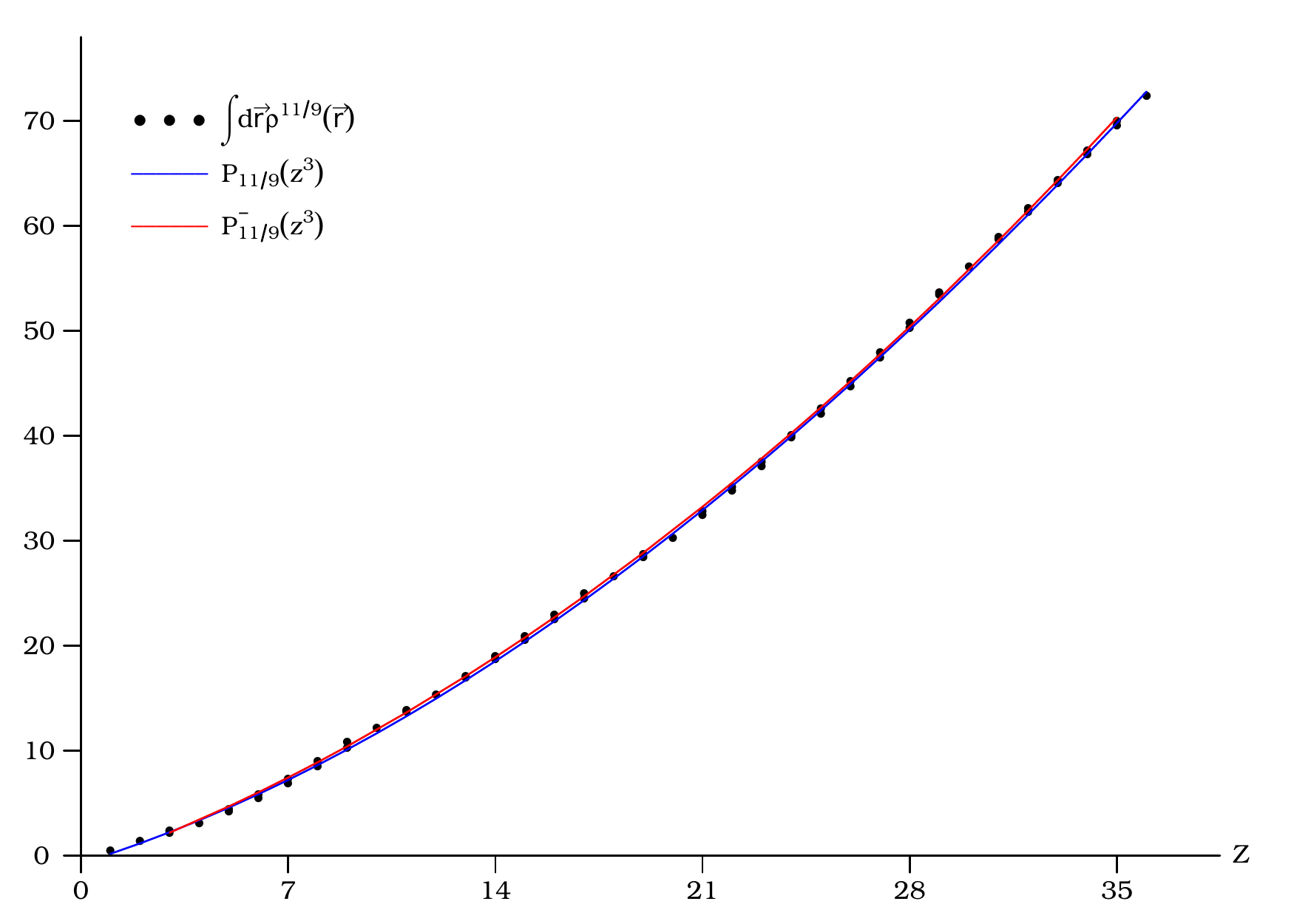}\\
		\hline
          \end{tabular}\\
          \vspace{0.5cm} (c)
        \end{minipage}
        \hfill{}
        \begin{minipage}[r]{1.0\columnwidth}
           \centering\begin{tabular}{|p{8.5cm}|}
	      \hline
		\includegraphics[width = 8.5cm, height = 5.5cm]{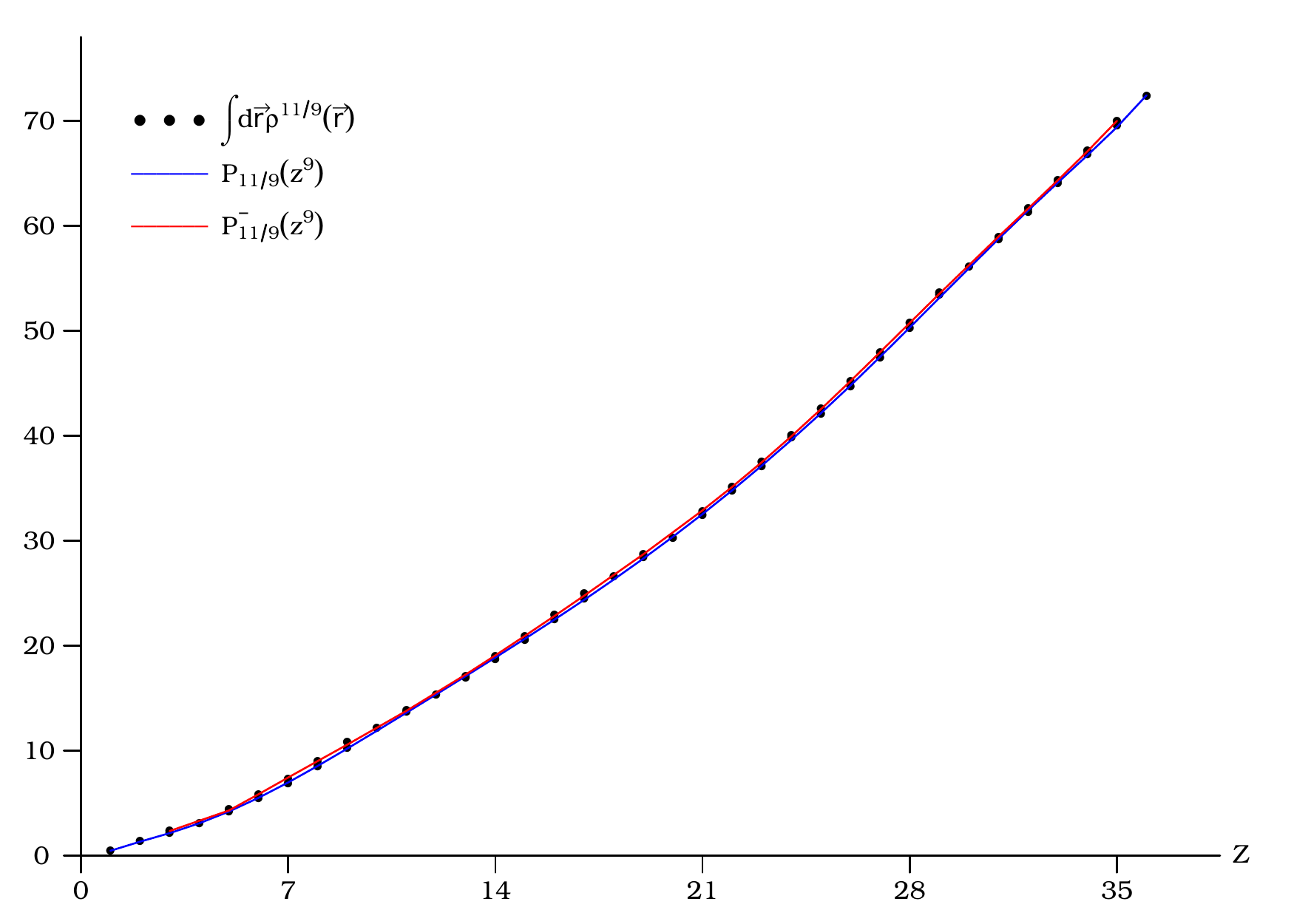}\\
		\hline
          \end{tabular}\\
          \vspace{0.5cm} (d)
        \end{minipage}
         \caption{Interpolation curves for the values of the $\int d\vec r\
\rho^{\nicefrac{11}{9}}(\vec r)$ through: (a) a 3rd degree polynomial  $P_{\nicefrac{11}{9}}(Z^3)$
(full blue) for 36 atoms and a 3rd degree polynomial  $P^-_{\nicefrac{11}{9}}(Z^3)$ (full red) for
27 negative ions, and (b) a 9th degree polynomial $P_{\nicefrac{11}{9}}(Z^9)$ (full blue) for 36
atoms and a 9th degree polynomial $P^-_{\nicefrac{11}{9}}(Z^9)$ (full red) for 27 negative
ions.}\label{fig:119n}
    \end{figure*}
    
    

These interpolation  polynomials are explicitly defined by:
\begin{align}
 P_{\nicefrac{4}{3}}(Z^3)  = &\ -0.9691803682 + 0.7854208699Z\nonumber\\
   &\  + 0.0776145852Z^2 - 0.0001581219Z^3,\label{eq:p1}
\end{align}
\begin{align}
  P_{\nicefrac{4}{3}}(Z^9)= &\ -1.0960551055 + 1.8518814624Z \nonumber\\
  &\ - 0.5991519550Z^2 + 0.1549675741Z^3\nonumber\\
  &\  -  0.0180687925Z^4 + 0.0012312619Z^5\nonumber\\
  &\ 	- 0.0000517284Z^6 + 0.0000013252Z^7\nonumber\\
  &\  - 0.0000000190Z^8 + 0.0000000001Z^9\label{eq:p1.1}
\end{align}
\begin{align}
 P_{\nicefrac{11}{9}}(Z^3)  = &\ -0.7540383360 + 0.8813316184Z\nonumber\\
  &\ +	0.0373453207Z^2 - 0.0001408691Z^3,\label{eq:p2}
\end{align}
\begin{align}
P_{\nicefrac{11}{9}}(Z^9)  = &\ -0.8077949490 + 1.6355990588Z\nonumber\\
  &\ - 0.4837629283Z^2 + 0.1255298989Z^3\nonumber\\
  &\ - 0.0150967704Z^4 + 0.0010441963Z^5\nonumber\\
  &\ - 0.0000439707Z^6 + 0.0000011196Z^7\nonumber\\
  &\ - 0.0000000159Z^8 + 0.0000000001Z^9\label{eq:p2.1}
\end{align}
where the coefficients are determined by least-square fitting.
Thus, the enhancement factor takes the following form (where $n$ is the degree of the $Z$ polynomial):
\begin{eqnarray}
 {{1}\over{2}} A_{Z^n,Appr}[\rho,Z] &=&  C_{T_1} +
  C_{T_2}\rho^{-\nicefrac{1}{3}}(\vec r)P_{\nicefrac{4}{3}}(Z^n)\nonumber\\
  & & +C_{T_3}\rho^{-\nicefrac{4}{9}}(\vec r)P_{\nicefrac{11}{9}}(Z^n)^2\nonumber\\
&&+\lambda\frac{\nabla^2\rho(\vec r)}{\rho(\vec
r)^{\nicefrac{5}{3}}}\nonumber\\
&&- \frac{1}{8}\frac{|\nabla\rho(\vec r)|^2}{\rho^{\nicefrac{8}{3}}(\vec
r)}\label{eq:6.3}
\end{eqnarray}

Also, this leads to the following approximation for the non-interacting kinetic energy functional
\begin{eqnarray}
  T_{LP97 + Z^n}[\rho,Z] &=& C_{T_1}\int d\vec r \rho^{\nicefrac{5}{3}}(\vec
r)\nonumber\\
  & & +C_{T_2}\int d\vec r \rho^{\nicefrac{4}{3}}(\vec r)P_{\nicefrac{4}{3}}(Z^n)\nonumber\\
 & &  +C_{T_3}\int d\vec r \rho^{\nicefrac{11}{9}}(\vec r)P_{\nicefrac{11}{9}}(Z^n)^2
 \label{eq:7.3}
\end{eqnarray}



\subsection{Application to neutral atoms}

The kinetic energy values corresponding to these $Z$-$\lambda$-dependent  functionals $T_{LP97 + Z^3}$ and  $T_{LP97 + Z^9}$ evaluated both with the Liu and Parr  and with newly optimized coefficients are presented in Table
\ref{Tab:0.4}.   Also in this table the values of the  Liu and Parr \cite{3.4} functional $T_{LP97}$ and the values of $T_{HF}$ reported by Clementi and Roetti,\cite{3.6}
are presented for comparison purposes. The  percentage relative errors are taken with respect to
$T_{HF}$.

The graphs  of these  new $Z$-$\lambda$-dependent  enhancement factors  (\ref{eq:6.3})	 are also plotted  in  Figs.  [\ref{fig:FactorNaApr1}] through
[\ref{fig:FactorKrApr3}].  Let us note that the graphs corresponding to these enhancement factors (dotted orange) coincide with those of the locally adjusted $\lambda$-dependent factors (dotted green) and hence they  are undistinguishable  in these figures.

\begin{table*}[!htbp]

\begin{center}
 \centering
  \caption{Non-interacting kinetic energy values for neutral atoms corresponding to the functionals  $T_{LP97}$, $T_{LP97 + Z^3}$, $T_{LP97 + Z^9}$,     and $T_{HF}$.}
\label{Tab:0.4}
  \smallskip
  \begin{minipage}{18cm}
\smallskip
\scalebox{0.95}[0.90]{
\begin{tabular}{lrrrrrr}
\hline\hline
  Atoms & {\hspace{0.23cm}$T_{LP97}$\footnote{$T_{LP97}$, Eq. (9), with Liu-Parr coefficients}} (error\%)& 
  {\hspace{0.23cm}$T_{LP97 + Z^3}$\footnote{$T_{LP97 + Z^3}$, Eq. (18), with reoptimized
coefficients ($C_{T_1}$ = 3.1336517827, $C_{T_2}$ = -0.0043445677 and $C_{T_3}$ = -0.0000345496)}}
(error\%)& {\hspace{0.23cm}$T_{LP97 + Z^9}$\footnote{$T_{LP97 + Z^9}$, Eq. (18), with
reoptimized coefficients ($C_{T_1}$ = 3.1257333712, $C_{T_2}$ = -0.0030202454 and $C_{T_3}$ =
-0.0000669074)}} (error\%)& {\hspace{0.23cm}$T_{LP97 + Z^3}$\footnote{$T_{LP97 + Z^3}$, Eq. (18), with  Liu-Parr coefficients}} (error\%)& {\hspace{0.23cm}$T_{LP97 + Z^9}$\footnote{$T_{LP97 + Z^9}$, Eq. (18), with Liu-Parr coefficients}} (error\%)& {\hspace{0.3cm}$T_{HF}$\footnote{$T_{HF}$, Eq. (2), reported by Clementi and Roetti\cite{3.6}}}\\
  \hline
  H & 0.327 (34.600 \hspace{-0.3em}) & 0.316 (36.800 \hspace{-0.3em}) & 0.314 (37.200 \hspace{-0.3em}) & 0.328 (34.400 \hspace{-0.3em}) & 0.326 (34.800 \hspace{-0.3em}) & 0.500\\ \noalign{\smallskip} 
He & 2.875 (0.454 \hspace{-0.3em}) & 2.791 (2.481 \hspace{-0.3em}) & 2.783 (2.760 \hspace{-0.3em}) & 2.890 (0.978 \hspace{-0.3em}) & 2.875 (0.454 \hspace{-0.3em}) & 2.862\\ \noalign{\smallskip} 
Li & 7.487 (0.726 \hspace{-0.3em}) & 7.271 (2.179 \hspace{-0.3em}) & 7.258 (2.354 \hspace{-0.3em}) & 7.485 (0.700 \hspace{-0.3em}) & 7.488 (0.740 \hspace{-0.3em}) & 7.433\\ \noalign{\smallskip} 
Be & 14.682 (0.748 \hspace{-0.3em}) & 14.277 (2.031 \hspace{-0.3em}) & 14.261 (2.141 \hspace{-0.3em}) & 14.639 (0.453 \hspace{-0.3em}) & 14.689 (0.796 \hspace{-0.3em}) & 14.573\\ \noalign{\smallskip} 
B & 24.496 (0.135 \hspace{-0.3em}) & 23.873 (2.674 \hspace{-0.3em}) & 23.856 (2.744 \hspace{-0.3em}) & 24.401 (0.522 \hspace{-0.3em}) & 24.507 (0.090 \hspace{-0.3em}) & 24.529\\ \noalign{\smallskip} 
C & 37.400 (0.764 \hspace{-0.3em}) & 36.592 (2.908 \hspace{-0.3em}) & 36.570 (2.966 \hspace{-0.3em}) & 37.302 (1.024 \hspace{-0.3em}) & 37.450 (0.632 \hspace{-0.3em}) & 37.688\\ \noalign{\smallskip} 
N & 53.852 (1.009 \hspace{-0.3em}) & 52.761 (3.015 \hspace{-0.3em}) & 52.728 (3.075 \hspace{-0.3em}) & 53.666 (1.351 \hspace{-0.3em}) & 53.819 (1.070 \hspace{-0.3em}) & 54.401\\ \noalign{\smallskip} 
O & 74.165 (0.861 \hspace{-0.3em}) & 72.912 (2.536 \hspace{-0.3em}) & 72.859 (2.607 \hspace{-0.3em}) & 74.025 (1.048 \hspace{-0.3em}) & 74.139 (0.896 \hspace{-0.3em}) & 74.809\\ \noalign{\smallskip} 
F & 98.982 (0.430 \hspace{-0.3em}) & 97.731 (1.688 \hspace{-0.3em}) & 97.649 (1.770 \hspace{-0.3em}) & 99.073 (0.338 \hspace{-0.3em}) & 99.105 (0.306 \hspace{-0.3em}) & 99.409\\ \noalign{\smallskip} 
Ne & 128.900 (0.275 \hspace{-0.3em}) & 127.557 (0.770 \hspace{-0.3em}) & 127.438 (0.863 \hspace{-0.3em}) & 129.144 (0.464 \hspace{-0.3em}) & 129.065 (0.403 \hspace{-0.3em}) & 128.547\\ \noalign{\smallskip} 
Na & 162.550 (0.428 \hspace{-0.3em}) & 161.064 (0.490 \hspace{-0.3em}) & 160.909 (0.586 \hspace{-0.3em}) & 162.838 (0.606 \hspace{-0.3em}) & 162.635 (0.481 \hspace{-0.3em}) & 161.857\\ \noalign{\smallskip} 
Mg & 200.660 (0.526 \hspace{-0.3em}) & 199.075 (0.268 \hspace{-0.3em}) & 198.886 (0.363 \hspace{-0.3em}) & 200.992 (0.692 \hspace{-0.3em}) & 200.674 (0.533 \hspace{-0.3em}) & 199.610\\ \noalign{\smallskip} 
Al & 243.190 (0.545 \hspace{-0.3em}) & 241.528 (0.142 \hspace{-0.3em}) & 241.315 (0.230 \hspace{-0.3em}) & 243.527 (0.684 \hspace{-0.3em}) & 243.120 (0.516 \hspace{-0.3em}) & 241.872\\ \noalign{\smallskip} 
Si & 290.360 (0.523 \hspace{-0.3em}) & 288.697 (0.052 \hspace{-0.3em}) & 288.472 (0.130 \hspace{-0.3em}) & 290.713 (0.646 \hspace{-0.3em}) & 290.254 (0.487 \hspace{-0.3em}) & 288.848\\ \noalign{\smallskip} 
P & 342.360 (0.483 \hspace{-0.3em}) & 340.741 (0.008 \hspace{-0.3em}) & 340.522 (0.056 \hspace{-0.3em}) & 342.702 (0.583 \hspace{-0.3em}) & 342.240 (0.448 \hspace{-0.3em}) & 340.714\\ \noalign{\smallskip} 
S & 399.390 (0.475 \hspace{-0.3em}) & 397.812 (0.078 \hspace{-0.3em}) & 397.617 (0.029 \hspace{-0.3em}) & 399.643 (0.539 \hspace{-0.3em}) & 399.228 (0.434 \hspace{-0.3em}) & 397.502\\ \noalign{\smallskip} 
Cl & 461.400 (0.421 \hspace{-0.3em}) & 460.189 (0.158 \hspace{-0.3em}) & 460.037 (0.125 \hspace{-0.3em}) & 461.816 (0.512 \hspace{-0.3em}) & 461.499 (0.443 \hspace{-0.3em}) & 459.464\\ \noalign{\smallskip} 
Ar & 528.920 (0.400 \hspace{-0.3em}) & 528.019 (0.229 \hspace{-0.3em}) & 527.929 (0.212 \hspace{-0.3em}) & 529.371 (0.485 \hspace{-0.3em}) & 529.199 (0.453 \hspace{-0.3em}) & 526.814\\ \noalign{\smallskip} 
K & 601.320 (0.360 \hspace{-0.3em}) & 600.499 (0.223 \hspace{-0.3em}) & 600.489 (0.221 \hspace{-0.3em}) & 601.470 (0.385 \hspace{-0.3em}) & 601.478 (0.386 \hspace{-0.3em}) & 599.164\\ \noalign{\smallskip} 
Ca & 678.810 (0.304 \hspace{-0.3em}) & 678.040 (0.190 \hspace{-0.3em}) & 678.123 (0.202 \hspace{-0.3em}) & 678.540 (0.264 \hspace{-0.3em}) & 678.750 (0.295 \hspace{-0.3em}) & 676.756\\ \noalign{\smallskip} 
Sc & 761.070 (0.176 \hspace{-0.3em}) & 760.702 (0.128 \hspace{-0.3em}) & 760.883 (0.151 \hspace{-0.3em}) & 760.650 (0.121 \hspace{-0.3em}) & 761.066 (0.175 \hspace{-0.3em}) & 759.733\\ \noalign{\smallskip} 
Ti & 848.790 (0.045 \hspace{-0.3em}) & 849.099 (0.081 \hspace{-0.3em}) & 849.372 (0.114 \hspace{-0.3em}) & 848.447 (0.005 \hspace{-0.3em}) & 849.047 (0.075 \hspace{-0.3em}) & 848.408\\ \noalign{\smallskip} 
V & 942.190 (0.073 \hspace{-0.3em}) & 943.030 (0.016 \hspace{-0.3em}) & 943.378 (0.052 \hspace{-0.3em}) & 941.736 (0.122 \hspace{-0.3em}) & 942.473 (0.043 \hspace{-0.3em}) & 942.883\\ \noalign{\smallskip} 
Cr & 1041.300 (0.196 \hspace{-0.3em}) & 1042.751 (0.057 \hspace{-0.3em}) & 1043.144 (0.020 \hspace{-0.3em}) & 1040.800 (0.244 \hspace{-0.3em}) & 1041.604 (0.167 \hspace{-0.3em}) & 1043.348\\ \noalign{\smallskip} 
Mn & 1147.000 (0.248 \hspace{-0.3em}) & 1148.759 (0.095 \hspace{-0.3em}) & 1149.148 (0.062 \hspace{-0.3em}) & 1146.177 (0.320 \hspace{-0.3em}) & 1146.958 (0.252 \hspace{-0.3em}) & 1149.857\\ \noalign{\smallskip} 
Fe & 1258.600 (0.295 \hspace{-0.3em}) & 1261.033 (0.103 \hspace{-0.3em}) & 1261.360 (0.077 \hspace{-0.3em}) & 1257.872 (0.353 \hspace{-0.3em}) & 1258.532 (0.301 \hspace{-0.3em}) & 1262.328\\ \noalign{\smallskip} 
Co & 1376.900 (0.326 \hspace{-0.3em}) & 1380.061 (0.097 \hspace{-0.3em}) & 1380.260 (0.083 \hspace{-0.3em}) & 1376.425 (0.361 \hspace{-0.3em}) & 1376.874 (0.328 \hspace{-0.3em}) & 1381.405\\ \noalign{\smallskip} 
Ni & 1501.900 (0.329 \hspace{-0.3em}) & 1505.871 (0.066 \hspace{-0.3em}) & 1505.884 (0.065 \hspace{-0.3em}) & 1501.900 (0.329 \hspace{-0.3em}) & 1502.075 (0.318 \hspace{-0.3em}) & 1506.862\\ \noalign{\smallskip} 
Cu & 1634.400 (0.278 \hspace{-0.3em}) & 1638.995 (0.002 \hspace{-0.3em}) & 1638.785 (0.011 \hspace{-0.3em}) & 1634.893 (0.248 \hspace{-0.3em}) & 1634.776 (0.255 \hspace{-0.3em}) & 1638.960\\ \noalign{\smallskip} 
Zn & 1773.600 (0.239 \hspace{-0.3em}) & 1778.280 (0.025 \hspace{-0.3em}) & 1777.860 (0.001 \hspace{-0.3em}) & 1774.250 (0.202 \hspace{-0.3em}) & 1773.877 (0.223 \hspace{-0.3em}) & 1777.842\\ \noalign{\smallskip} 
Ga & 1919.400 (0.200 \hspace{-0.3em}) & 1923.884 (0.033 \hspace{-0.3em}) & 1923.330 (0.004 \hspace{-0.3em}) & 1920.187 (0.159 \hspace{-0.3em}) & 1919.643 (0.188 \hspace{-0.3em}) & 1923.254\\ \noalign{\smallskip} 
Ge & 2072.700 (0.128 \hspace{-0.3em}) & 2075.981 (0.030 \hspace{-0.3em}) & 2075.442 (0.004 \hspace{-0.3em}) & 2072.940 (0.117 \hspace{-0.3em}) & 2072.346 (0.145 \hspace{-0.3em}) & 2075.360\\ \noalign{\smallskip} 
As & 2232.200 (0.091 \hspace{-0.3em}) & 2234.766 (0.024 \hspace{-0.3em}) & 2234.444 (0.009 \hspace{-0.3em}) & 2232.779 (0.065 \hspace{-0.3em}) & 2232.254 (0.089 \hspace{-0.3em}) & 2234.237\\ \noalign{\smallskip} 
Se & 2399.100 (0.032 \hspace{-0.3em}) & 2400.099 (0.010 \hspace{-0.3em}) & 2400.173 (0.013 \hspace{-0.3em}) & 2399.625 (0.010 \hspace{-0.3em}) & 2399.257 (0.025 \hspace{-0.3em}) & 2399.865\\ \noalign{\smallskip} 
Br & 2573.600 (0.045 \hspace{-0.3em}) & 2572.359 (0.003 \hspace{-0.3em}) & 2572.813 (0.015 \hspace{-0.3em}) & 2573.947 (0.059 \hspace{-0.3em}) & 2573.782 (0.052 \hspace{-0.3em}) & 2572.432\\ \noalign{\smallskip} 
Kr & 2755.700 (0.133 \hspace{-0.3em}) & 2751.565 (0.018 \hspace{-0.3em}) & 2751.863 (0.007 \hspace{-0.3em}) & 2755.845 (0.138 \hspace{-0.3em}) & 2755.967 (0.142 \hspace{-0.3em}) & 2752.048\\ \noalign{\smallskip} 

  \hline
  MAD & 0.222 \hspace{0.1em} & 0.129 \hspace{0.1em} & 0.166 \hspace{0.1em} & 0.329 \hspace{0.1em} & 0.224 \hspace{0.1em} & \\ \noalign{\smallskip} 

  \hline\hline
\end{tabular}}
\end{minipage}
\end{center}
\end{table*}

\begin{table*}[!htbp]

\begin{center}
 \centering
  \caption{Non-interacting kinetic energy values for positive ions  corresponding to the functionals  $T_{LP97}$,
$T_{LP97 + Z^3}$, $T_{LP97 + Z^9}$,     and $T_{HF}$.}
\label{Tab:0.5}
  \smallskip
  \begin{minipage}{18cm}
\smallskip
\scalebox{0.95}[0.90]{
\begin{tabular}{lrrrrrr}
\hline\hline
  Ions & {\hspace{0.23cm}$T_{LP97}$\footnote{$T_{LP97}$, Eq. (9), with Liu-Parr coefficients for neutral atoms }} (error\%)&
{\hspace{0.23cm}$T_{LP97}$\footnote{$T_{LP97}$, Eq. (9), with reoptimized
coefficients ($C_{T_1}$ =  3.1288539558, $C_{T_2}$ = -0.0034574267 and $C_{T_3}$ = -0.0000591469)}}
(error\%)&{\hspace{0.23cm}$T_{LP97 + Z^9}$\footnote{$T_{LP97 + Z^9}$, Eq. (18), with both $Z^9$ and  Liu-Parr coefficients for neutral atoms}} (error\%)&{\hspace{0.23cm}$T_{LP97 + Z^9}$\footnote{$T_{LP97 + Z^9}$, Eq. (18), with $Z^9$ for neutral atoms and 
reoptimized coefficients ($C_{T_1}$ = 3.1267059586, $C_{T_2}$ = -0.0039716465 and $C_{T_3}$ =
-0.0000518778)}} (error\%)&{\hspace{0.23cm}$T_{LP97 + Z^3}$\footnote{$T_{LP97 + Z^3}$, Eq. (18), with $Z^3$ for positive ions and  reoptimized
coefficients ($C_{T_1}$ = 3.1370019499, $C_{T_2}$ = -0.0048541396 and $C_{T_3}$ = -0.0000240227)}}
(error\%)& {\hspace{0.3cm}$T_{HF}$\footnote{$T_{HF}$, Eq. (2), reported
by  Clementi and Roetti\cite{3.6}}}\\
  \hline
  Li$^+$ & 7.346 (1.520 \hspace{-0.3em}) & 7.118 (1.631 \hspace{-0.3em}) & 7.334 (1.354 \hspace{-0.3em}) & 7.125 (1.534 \hspace{-0.3em}) & 7.140 (1.327 \hspace{-0.3em}) & 7.236\\ \noalign{\smallskip} 
Be$^+$ & 14.369 (0.637 \hspace{-0.3em}) & 13.945 (2.332 \hspace{-0.3em}) & 14.343 (0.455 \hspace{-0.3em}) & 13.956 (2.255 \hspace{-0.3em}) & 13.973 (2.136 \hspace{-0.3em}) & 14.278\\ \noalign{\smallskip} 
B$^+$ & 24.348 (0.458 \hspace{-0.3em}) & 23.678 (2.306 \hspace{-0.3em}) & 24.294 (0.235 \hspace{-0.3em}) & 23.695 (2.236 \hspace{-0.3em}) & 23.708 (2.183 \hspace{-0.3em}) & 24.237\\ \noalign{\smallskip} 
C$^+$ & 37.076 (0.579 \hspace{-0.3em}) & 36.150 (3.062 \hspace{-0.3em}) & 36.968 (0.869 \hspace{-0.3em}) & 36.166 (3.019 \hspace{-0.3em}) & 36.177 (2.990 \hspace{-0.3em}) & 37.292\\ \noalign{\smallskip} 
N$^+$ & 53.345 (1.008 \hspace{-0.3em}) & 52.164 (3.199 \hspace{-0.3em}) & 53.166 (1.340 \hspace{-0.3em}) & 52.178 (3.173 \hspace{-0.3em}) & 52.192 (3.147 \hspace{-0.3em}) & 53.888\\ \noalign{\smallskip} 
O$^+$ & 73.401 (1.306 \hspace{-0.3em}) & 71.990 (3.203 \hspace{-0.3em}) & 73.143 (1.653 \hspace{-0.3em}) & 72.000 (3.189 \hspace{-0.3em}) & 72.024 (3.157 \hspace{-0.3em}) & 74.372\\ \noalign{\smallskip} 
F$^+$ & 97.889 (0.953 \hspace{-0.3em}) & 96.302 (2.559 \hspace{-0.3em}) & 97.592 (1.254 \hspace{-0.3em}) & 96.313 (2.548 \hspace{-0.3em}) & 96.357 (2.503 \hspace{-0.3em}) & 98.831\\ \noalign{\smallskip} 
Ne$^+$ & 127.498 (0.249 \hspace{-0.3em}) & 125.829 (1.555 \hspace{-0.3em}) & 127.282 (0.418 \hspace{-0.3em}) & 125.866 (1.526 \hspace{-0.3em}) & 125.936 (1.471 \hspace{-0.3em}) & 127.816\\ \noalign{\smallskip} 
Na$^+$ & 162.437 (0.469 \hspace{-0.3em}) & 160.810 (0.537 \hspace{-0.3em}) & 162.470 (0.490 \hspace{-0.3em}) & 160.910 (0.475 \hspace{-0.3em}) & 161.010 (0.413 \hspace{-0.3em}) & 161.678\\ \noalign{\smallskip} 
Mg$^+$ & 200.463 (0.548 \hspace{-0.3em}) & 198.700 (0.336 \hspace{-0.3em}) & 200.374 (0.504 \hspace{-0.3em}) & 198.778 (0.297 \hspace{-0.3em}) & 198.905 (0.233 \hspace{-0.3em}) & 199.370\\ \noalign{\smallskip} 
Al$^+$ & 243.131 (0.603 \hspace{-0.3em}) & 241.285 (0.161 \hspace{-0.3em}) & 242.945 (0.526 \hspace{-0.3em}) & 241.348 (0.134 \hspace{-0.3em}) & 241.492 (0.075 \hspace{-0.3em}) & 241.673\\ \noalign{\smallskip} 
Si$^+$ & 290.220 (0.573 \hspace{-0.3em}) & 288.343 (0.077 \hspace{-0.3em}) & 289.938 (0.475 \hspace{-0.3em}) & 288.389 (0.061 \hspace{-0.3em}) & 288.540 (0.009 \hspace{-0.3em}) & 288.566\\ \noalign{\smallskip} 
P$^+$ & 342.124 (0.523 \hspace{-0.3em}) & 340.297 (0.014 \hspace{-0.3em}) & 341.803 (0.429 \hspace{-0.3em}) & 340.342 (0.001 \hspace{-0.3em}) & 340.485 (0.041 \hspace{-0.3em}) & 340.344\\ \noalign{\smallskip} 
S$^+$ & 398.990 (0.458 \hspace{-0.3em}) & 397.291 (0.031 \hspace{-0.3em}) & 398.678 (0.380 \hspace{-0.3em}) & 397.349 (0.045 \hspace{-0.3em}) & 397.465 (0.075 \hspace{-0.3em}) & 397.169\\ \noalign{\smallskip} 
Cl$^+$ & 460.947 (0.415 \hspace{-0.3em}) & 459.465 (0.092 \hspace{-0.3em}) & 460.707 (0.362 \hspace{-0.3em}) & 459.554 (0.111 \hspace{-0.3em}) & 459.627 (0.127 \hspace{-0.3em}) & 459.044\\ \noalign{\smallskip} 
Ar$^+$ & 528.301 (0.386 \hspace{-0.3em}) & 527.142 (0.166 \hspace{-0.3em}) & 528.228 (0.372 \hspace{-0.3em}) & 527.291 (0.194 \hspace{-0.3em}) & 527.305 (0.196 \hspace{-0.3em}) & 526.271\\ \noalign{\smallskip} 
K$^+$ & 601.150 (0.357 \hspace{-0.3em}) & 600.404 (0.233 \hspace{-0.3em}) & 601.311 (0.384 \hspace{-0.3em}) & 600.637 (0.271 \hspace{-0.3em}) & 600.579 (0.262 \hspace{-0.3em}) & 599.011\\ \noalign{\smallskip} 
Ca$^+$ & 678.633 (0.305 \hspace{-0.3em}) & 678.139 (0.232 \hspace{-0.3em}) & 678.559 (0.294 \hspace{-0.3em}) & 678.256 (0.249 \hspace{-0.3em}) & 678.118 (0.229 \hspace{-0.3em}) & 676.569\\ \noalign{\smallskip} 
Sc$^+$ & 761.309 (0.244 \hspace{-0.3em}) & 761.077 (0.213 \hspace{-0.3em}) & 760.885 (0.188 \hspace{-0.3em}) & 761.032 (0.207 \hspace{-0.3em}) & 760.813 (0.178 \hspace{-0.3em}) & 759.459\\ \noalign{\smallskip} 
Ti$^+$ & 848.965 (0.107 \hspace{-0.3em}) & 849.283 (0.145 \hspace{-0.3em}) & 848.552 (0.059 \hspace{-0.3em}) & 849.221 (0.138 \hspace{-0.3em}) & 848.929 (0.103 \hspace{-0.3em}) & 848.054\\ \noalign{\smallskip} 
V$^+$ & 942.467 (0.003 \hspace{-0.3em}) & 943.441 (0.100 \hspace{-0.3em}) & 942.141 (0.038 \hspace{-0.3em}) & 943.380 (0.094 \hspace{-0.3em}) & 943.039 (0.058 \hspace{-0.3em}) & 942.496\\ \noalign{\smallskip} 
Cr$^+$ & 1041.397 (0.143 \hspace{-0.3em}) & 1043.016 (0.012 \hspace{-0.3em}) & 1041.026 (0.178 \hspace{-0.3em}) & 1042.902 (0.001 \hspace{-0.3em}) & 1042.546 (0.033 \hspace{-0.3em}) & 1042.887\\ \noalign{\smallskip} 
Mn$^+$ & 1146.500 (0.249 \hspace{-0.3em}) & 1148.804 (0.048 \hspace{-0.3em}) & 1146.107 (0.283 \hspace{-0.3em}) & 1148.631 (0.063 \hspace{-0.3em}) & 1148.312 (0.091 \hspace{-0.3em}) & 1149.359\\ \noalign{\smallskip} 
Fe$^+$ & 1258.002 (0.325 \hspace{-0.3em}) & 1261.002 (0.087 \hspace{-0.3em}) & 1257.622 (0.355 \hspace{-0.3em}) & 1260.772 (0.105 \hspace{-0.3em}) & 1260.548 (0.123 \hspace{-0.3em}) & 1262.102\\ \noalign{\smallskip} 
Co$^+$ & 1376.016 (0.360 \hspace{-0.3em}) & 1379.581 (0.102 \hspace{-0.3em}) & 1375.505 (0.397 \hspace{-0.3em}) & 1379.216 (0.128 \hspace{-0.3em}) & 1379.151 (0.133 \hspace{-0.3em}) & 1380.989\\ \noalign{\smallskip} 
Ni$^+$ & 1500.922 (0.365 \hspace{-0.3em}) & 1505.084 (0.088 \hspace{-0.3em}) & 1500.588 (0.387 \hspace{-0.3em}) & 1504.711 (0.113 \hspace{-0.3em}) & 1504.858 (0.103 \hspace{-0.3em}) & 1506.416\\ \noalign{\smallskip} 
Cu$^+$ & 1632.914 (0.339 \hspace{-0.3em}) & 1637.470 (0.061 \hspace{-0.3em}) & 1632.779 (0.347 \hspace{-0.3em}) & 1637.113 (0.083 \hspace{-0.3em}) & 1637.498 (0.059 \hspace{-0.3em}) & 1638.468\\ \noalign{\smallskip} 
Zn$^+$ & 1772.639 (0.276 \hspace{-0.3em}) & 1777.364 (0.011 \hspace{-0.3em}) & 1773.437 (0.232 \hspace{-0.3em}) & 1777.679 (0.007 \hspace{-0.3em}) & 1778.277 (0.041 \hspace{-0.3em}) & 1777.553\\ \noalign{\smallskip} 
Ga$^+$ & 1918.639 (0.229 \hspace{-0.3em}) & 1923.234 (0.010 \hspace{-0.3em}) & 1919.421 (0.189 \hspace{-0.3em}) & 1923.362 (0.016 \hspace{-0.3em}) & 1924.076 (0.054 \hspace{-0.3em}) & 1923.047\\ \noalign{\smallskip} 
Ge$^+$ & 2071.275 (0.184 \hspace{-0.3em}) & 2075.501 (0.020 \hspace{-0.3em}) & 2072.113 (0.143 \hspace{-0.3em}) & 2075.479 (0.019 \hspace{-0.3em}) & 2076.137 (0.051 \hspace{-0.3em}) & 2075.085\\ \noalign{\smallskip} 
As$^+$ & 2230.704 (0.137 \hspace{-0.3em}) & 2234.199 (0.019 \hspace{-0.3em}) & 2231.664 (0.094 \hspace{-0.3em}) & 2234.165 (0.018 \hspace{-0.3em}) & 2234.536 (0.034 \hspace{-0.3em}) & 2233.766\\ \noalign{\smallskip} 
Se$^+$ & 2397.498 (0.080 \hspace{-0.3em}) & 2399.722 (0.013 \hspace{-0.3em}) & 2398.622 (0.033 \hspace{-0.3em}) & 2399.891 (0.020 \hspace{-0.3em}) & 2399.765 (0.014 \hspace{-0.3em}) & 2399.418\\ \noalign{\smallskip} 
Br$^+$ & 2571.623 (0.010 \hspace{-0.3em}) & 2572.071 (0.007 \hspace{-0.3em}) & 2572.959 (0.041 \hspace{-0.3em}) & 2572.410 (0.020 \hspace{-0.3em}) & 2571.761 (0.005 \hspace{-0.3em}) & 2571.893\\ \noalign{\smallskip} 
Kr$^+$ & 2753.535 (0.072 \hspace{-0.3em}) & 2751.537 (0.001 \hspace{-0.3em}) & 2755.166 (0.131 \hspace{-0.3em}) & 2751.572 (0.000 \hspace{-0.3em}) & 2750.878 (0.025 \hspace{-0.3em}) & 2751.567\\ \noalign{\smallskip} 

  \hline
  MAD & 0.251 \hspace{0.1em} & 0.133 \hspace{0.1em} & 0.202 \hspace{0.1em} & 0.153 \hspace{0.1em} & 0.121 \hspace{0.1em} & \\ \noalign{\smallskip} 

  \hline\hline
\end{tabular}}
\end{minipage}
\end{center}
\end{table*}

\begin{table*}[!htbp]

\begin{center}
 \centering
  \caption{Non-interacting kinetic energy values for negative ions corresponding to the functionals  $T_{LP97}$,
$T_{LP97 + Z^3}$, $T_{LP97 + Z^9}$,     and $T_{HF}$.}
\label{Tab:0.7}
  \smallskip
  \begin{minipage}{18cm}
\smallskip
\scalebox{0.95}[0.90]{
\begin{tabular}{lrrrrrr}
\hline\hline
  Ions & {\hspace{0.23cm}$T_{LP97}$\footnote{$T_{LP97}$, Eq. (9) with Liu-Parr coefficients for neutral atoms}} (error\%)&
{\hspace{0.23cm}$T_{LP97}$\footnote{$T_{LP97}$, Eq. (9) with reoptimized
coefficients ($C_{T_1}$ = 3.1248445957, $C_{T_2}$ = -0.0027038794 and $C_{T_3}$ = -0.0000764926)}}
(error\%)&{\hspace{0.23cm}$T_{LP97 + Z^9}$\footnote{$T_{LP97 + Z^9}$, Eq. (18) with both  $Z^9$ and  Liu-Parr coefficients for neutral atoms}} (error\%)&{\hspace{0.23cm}$T_{LP97 + Z^9}$\footnote{$T_{LP97 + Z^9}$, Eq. (18) with $Z^9$ for neutral atoms and 
reoptimized coefficients ($C_{T_1}$ = 3.1218035155, $C_{T_2}$ = -0.0022915847 and $C_{T_3}$ =
-0.0000868695)}} (error\%)&{\hspace{0.23cm}$T_{LP97 + Z^3}$\footnote{$T_{LP97 + Z^3}$, Eq. (18) with $Z^3$ for negative ions and  reoptimized coefficients ($C_{T_1}$ = 3.1326163517, $C_{T_2}$ = -0.0040144747 and $C_{T_3}$ = -0.0000443462)}}
(error\%)& {\hspace{0.3cm}$T_{HF}$\footnote{Values of $T_{HF}$ reported
by  Clementi and Roetti\cite{3.6}}}\\
  \hline
  Li$^-$ & 7.504 (1.023 \hspace{-0.3em}) & 7.280 (1.992 \hspace{-0.3em}) & 7.512 (1.131 \hspace{-0.3em}) & 7.286 (1.912 \hspace{-0.3em}) & 7.298 (1.750 \hspace{-0.3em}) & 7.428\\ \noalign{\smallskip} 
B$^-$ & 24.386 (0.384 \hspace{-0.3em}) & 23.759 (2.945 \hspace{-0.3em}) & 24.411 (0.282 \hspace{-0.3em}) & 23.779 (2.864 \hspace{-0.3em}) & 23.784 (2.843 \hspace{-0.3em}) & 24.480\\ \noalign{\smallskip} 
C$^-$ & 37.474 (0.338 \hspace{-0.3em}) & 36.622 (2.604 \hspace{-0.3em}) & 37.513 (0.234 \hspace{-0.3em}) & 36.650 (2.529 \hspace{-0.3em}) & 36.651 (2.527 \hspace{-0.3em}) & 37.601\\ \noalign{\smallskip} 
N$^-$ & 53.835 (0.650 \hspace{-0.3em}) & 52.764 (2.626 \hspace{-0.3em}) & 53.865 (0.594 \hspace{-0.3em}) & 52.797 (2.565 \hspace{-0.3em}) & 52.799 (2.561 \hspace{-0.3em}) & 54.187\\ \noalign{\smallskip} 
O$^-$ & 74.512 (0.377 \hspace{-0.3em}) & 73.283 (2.020 \hspace{-0.3em}) & 74.599 (0.261 \hspace{-0.3em}) & 73.333 (1.953 \hspace{-0.3em}) & 73.344 (1.939 \hspace{-0.3em}) & 74.794\\ \noalign{\smallskip} 
F$^-$ & 99.489 (0.030 \hspace{-0.3em}) & 98.161 (1.305 \hspace{-0.3em}) & 99.677 (0.219 \hspace{-0.3em}) & 98.237 (1.229 \hspace{-0.3em}) & 98.265 (1.200 \hspace{-0.3em}) & 99.459\\ \noalign{\smallskip} 
Na$^-$ & 162.551 (0.431 \hspace{-0.3em}) & 160.922 (0.575 \hspace{-0.3em}) & 162.655 (0.496 \hspace{-0.3em}) & 160.981 (0.539 \hspace{-0.3em}) & 161.058 (0.491 \hspace{-0.3em}) & 161.853\\ \noalign{\smallskip} 
Al$^-$ & 243.169 (0.556 \hspace{-0.3em}) & 241.344 (0.198 \hspace{-0.3em}) & 243.099 (0.527 \hspace{-0.3em}) & 241.361 (0.191 \hspace{-0.3em}) & 241.478 (0.143 \hspace{-0.3em}) & 241.824\\ \noalign{\smallskip} 
Si$^-$ & 290.362 (0.537 \hspace{-0.3em}) & 288.550 (0.090 \hspace{-0.3em}) & 290.286 (0.511 \hspace{-0.3em}) & 288.573 (0.082 \hspace{-0.3em}) & 288.697 (0.039 \hspace{-0.3em}) & 288.810\\ \noalign{\smallskip} 
P$^-$ & 342.409 (0.502 \hspace{-0.3em}) & 340.687 (0.003 \hspace{-0.3em}) & 342.380 (0.494 \hspace{-0.3em}) & 340.732 (0.010 \hspace{-0.3em}) & 340.848 (0.044 \hspace{-0.3em}) & 340.697\\ \noalign{\smallskip} 
S$^-$ & 399.435 (0.477 \hspace{-0.3em}) & 397.873 (0.084 \hspace{-0.3em}) & 399.496 (0.493 \hspace{-0.3em}) & 397.955 (0.105 \hspace{-0.3em}) & 398.048 (0.128 \hspace{-0.3em}) & 397.538\\ \noalign{\smallskip} 
Cl$^-$ & 461.579 (0.435 \hspace{-0.3em}) & 460.266 (0.149 \hspace{-0.3em}) & 461.802 (0.483 \hspace{-0.3em}) & 460.410 (0.181 \hspace{-0.3em}) & 460.464 (0.192 \hspace{-0.3em}) & 459.580\\ \noalign{\smallskip} 
K$^-$ & 601.426 (0.378 \hspace{-0.3em}) & 600.595 (0.239 \hspace{-0.3em}) & 601.589 (0.405 \hspace{-0.3em}) & 600.681 (0.253 \hspace{-0.3em}) & 600.615 (0.242 \hspace{-0.3em}) & 599.164\\ \noalign{\smallskip} 
Sc$^-$ & 760.893 (0.158 \hspace{-0.3em}) & 760.787 (0.145 \hspace{-0.3em}) & 760.884 (0.157 \hspace{-0.3em}) & 760.784 (0.144 \hspace{-0.3em}) & 760.569 (0.116 \hspace{-0.3em}) & 759.689\\ \noalign{\smallskip} 
Ti$^-$ & 848.700 (0.038 \hspace{-0.3em}) & 849.124 (0.088 \hspace{-0.3em}) & 848.736 (0.042 \hspace{-0.3em}) & 849.139 (0.090 \hspace{-0.3em}) & 848.859 (0.057 \hspace{-0.3em}) & 848.377\\ \noalign{\smallskip} 
V$^-$ & 942.132 (0.077 \hspace{-0.3em}) & 943.120 (0.028 \hspace{-0.3em}) & 942.176 (0.072 \hspace{-0.3em}) & 943.145 (0.030 \hspace{-0.3em}) & 942.821 (0.004 \hspace{-0.3em}) & 942.859\\ \noalign{\smallskip} 
Cr$^-$ & 1041.659 (0.159 \hspace{-0.3em}) & 1043.267 (0.005 \hspace{-0.3em}) & 1041.757 (0.150 \hspace{-0.3em}) & 1043.323 (0.000 \hspace{-0.3em}) & 1042.990 (0.032 \hspace{-0.3em}) & 1043.320\\ \noalign{\smallskip} 
Mn$^-$ & 1147.102 (0.227 \hspace{-0.3em}) & 1149.267 (0.039 \hspace{-0.3em}) & 1147.103 (0.227 \hspace{-0.3em}) & 1149.297 (0.036 \hspace{-0.3em}) & 1149.002 (0.062 \hspace{-0.3em}) & 1149.713\\ \noalign{\smallskip} 
Fe$^-$ & 1258.973 (0.268 \hspace{-0.3em}) & 1261.672 (0.054 \hspace{-0.3em}) & 1258.887 (0.275 \hspace{-0.3em}) & 1261.682 (0.053 \hspace{-0.3em}) & 1261.484 (0.069 \hspace{-0.3em}) & 1262.355\\ \noalign{\smallskip} 
Co$^-$ & 1377.380 (0.286 \hspace{-0.3em}) & 1380.531 (0.058 \hspace{-0.3em}) & 1377.192 (0.300 \hspace{-0.3em}) & 1380.514 (0.059 \hspace{-0.3em}) & 1380.475 (0.062 \hspace{-0.3em}) & 1381.335\\ \noalign{\smallskip} 
Ni$^-$ & 1502.647 (0.277 \hspace{-0.3em}) & 1506.098 (0.048 \hspace{-0.3em}) & 1502.398 (0.293 \hspace{-0.3em}) & 1506.104 (0.047 \hspace{-0.3em}) & 1506.276 (0.036 \hspace{-0.3em}) & 1506.819\\ \noalign{\smallskip} 
Cu$^-$ & 1635.023 (0.240 \hspace{-0.3em}) & 1638.596 (0.022 \hspace{-0.3em}) & 1634.818 (0.253 \hspace{-0.3em}) & 1638.696 (0.016 \hspace{-0.3em}) & 1639.103 (0.009 \hspace{-0.3em}) & 1638.959\\ \noalign{\smallskip} 
Ga$^-$ & 1920.237 (0.156 \hspace{-0.3em}) & 1923.472 (0.012 \hspace{-0.3em}) & 1919.661 (0.186 \hspace{-0.3em}) & 1923.128 (0.006 \hspace{-0.3em}) & 1923.856 (0.032 \hspace{-0.3em}) & 1923.236\\ \noalign{\smallskip} 
Ge$^-$ & 2073.247 (0.101 \hspace{-0.3em}) & 2075.651 (0.015 \hspace{-0.3em}) & 2072.426 (0.141 \hspace{-0.3em}) & 2075.253 (0.004 \hspace{-0.3em}) & 2075.917 (0.028 \hspace{-0.3em}) & 2075.343\\ \noalign{\smallskip} 
As$^-$ & 2233.227 (0.040 \hspace{-0.3em}) & 2234.458 (0.015 \hspace{-0.3em}) & 2232.250 (0.084 \hspace{-0.3em}) & 2234.129 (0.000 \hspace{-0.3em}) & 2234.493 (0.016 \hspace{-0.3em}) & 2234.130\\ \noalign{\smallskip} 
Se$^-$ & 2400.688 (0.033 \hspace{-0.3em}) & 2400.011 (0.005 \hspace{-0.3em}) & 2399.491 (0.017 \hspace{-0.3em}) & 2400.042 (0.006 \hspace{-0.3em}) & 2399.888 (0.000 \hspace{-0.3em}) & 2399.897\\ \noalign{\smallskip} 
Br$^-$ & 2575.558 (0.118 \hspace{-0.3em}) & 2572.456 (0.003 \hspace{-0.3em}) & 2574.179 (0.064 \hspace{-0.3em}) & 2572.791 (0.010 \hspace{-0.3em}) & 2572.092 (0.017 \hspace{-0.3em}) & 2572.531\\ \noalign{\smallskip} 

  \hline
  MAD & 0.234 \hspace{0.1em} & 0.107 \hspace{0.1em} & 0.213 \hspace{0.1em} & 0.113 \hspace{0.1em} & 0.080 \hspace{0.1em} & \\ \noalign{\smallskip} 

  \hline\hline
\end{tabular}}
\end{minipage}
\end{center}
\end{table*}

\subsection{Application to positive and negative ions}

We have also examined  whether the approximate functional forms discussed above  are  applicable to  positive and negative ions.  For this purpose, we list in column 2 of Table II   the values of the non-interactive kinetic energy for positive ions calculated by means of the Liu-Parr expansion given by Eq. (9) using the same optimized coefficients  as those of neutral atoms but the corresponding densities of positive ions taken from the Clementi-Roetti tables.  Similar results are presented in column 2 of Table III for the case of negative ions.  In column 3 of Tables II and III, we present the evaluation of Eq. (9) for positive and negative ions, respectively, but where in each case the coefficients of the homogeneous functional expansion have been reoptimized.

In order to assess whether the $Z^n$ approximations  for neutral atoms can be transferred to positive and negative ions, we present in column 4 of Tables II and III,  respectively, the  values of the non-interactive energy computed by means of Eq. (18) where we have  used the $Z^9$ polynomial approximation fitted for the case of neutral atoms as well as the  Liu-Parr coefficients also optimized for neutral atoms. In column 5, of Tables II and III, the same results are presented for the case of the $Z^9$ polynomial approximation for neutral atoms but where the coefficients have been reoptimized. In column 6 of Tables II and III, we present values of the non-interacting kinetic energy calculated by means of Eq. (18), but where  the $Z^3$ polynomial approximation is  fitted for each particular case and where also the coefficients of the homogeneous functional expansion have been reoptimized. In the last column in Tables II  and III,  the Hartree-Fock values of Clementi and Roetti  for the non-interacting 
kinetic energy  of positive and negative ions are listed, respectively.

It is clearly seen from the values of columns 3 of Tables II and III that the Liu-Parr homogeneous functional expansion works quite well for  positive and negative atomic ions, respectively. In the case positive ions, the MAD value for the relative percent error is 0.133  ; for negative ions, 0.107.  In both of these cases, the accuracy increases for ions with large atomic number.  The transferability of  the $Z^n$ polynomial approximation for neutral atoms as well as the use of the Liu-Parr coefficients (optimized for neutral atoms) is examined in column 4 of Tables II and III. The results show a MAD value of  0.202 and 0.213 for positive and negative ions, respectively.  Again, we see that the approximation improves with increasing atomic number.  These results are, in fact, quite comparable to those of the original  Liu-Parr expression (Eq. (9)).  Naturally, the best fit is obtained when both the $Z^n$ function and the coefficients  have been optimized for the ions.  The MAD results in this case are 0.121 
and 0.080, respectively, for positive and negative ions.

\subsection{Extensions to molecular systems and clusters}

We base the present discussion on the fact, ascertained previously,\cite{1.5} that to a good approximation the kinetic energy  enhancement factor for two neighboring interacting  atoms is given by the sum of the atomic enhancement factors of the participating atoms.  Let us consider an electronic system, a molecule or a cluster which consists of M atoms.  We assume that the whole space can be divided into M subvolumes $\{\Omega_A\}_{A=1, . . . ,M}$ ,
each one of them corresponding to a given atom $A$. For, example, for the system composed by the same atoms the space
can be divided in the following way:
\begin{equation}
{\vec {r}}\in \Omega_A  \quad {\rm if} \quad {\rm min} \vert {\vec {r}}-{\vec {R}}_B\vert = \vert {\vec {r}}-{\vec {R}}_A\vert 
\end{equation}
where the vectors  ${\vec {R}}_A$ and ${\vec {R}}_B$ denote the nuclear positions.  Thus, 
${\mathbf{R}}^3=\bigcup_{A=1}\sp{M}\Omega_A$.

The one-particle density for a system  formed by $M$ atoms is $\rho({\vec {r}})\equiv \rho({\vec {r}}, \{{\vec {R}}_A\}_{A=1, . . . ,M})$, where we make explicit the presence of nuclear coordinates corresponding to the fixed atoms.
Let us define 
\begin{equation}
\rho_A({\vec {r}})\equiv \rho({\vec {r}}\in\Omega_A)
\label{}
\end{equation}
Clearly, the notion that $\rho_A({\vec { }r})$ is an atomic density is not implicit in this definition as the one-particle density $\rho({\vec {r}})$ corresponds to the whole system.  The definition of $\rho_A({\vec {r}})$ does imply, however, that the value of the molecular or cluster density is that associated  to a particular volume  $\Omega_A$.

Bearing in mind these considerations we see that we may write the second term of Eq. (1) as
\begin{eqnarray}
& &
\frac{1}{2}\int_{{\mathbf{R}}^3}
d\vec r\rho^{\nicefrac{5}{3}}(\vec r)A_N[\rho(\vec r);\vec r]\nonumber \\ &=&
\frac{1}{2}\sum_{A=1}\sp{M}\int_{\Omega_A}
d\vec r\rho_{A}^{\nicefrac{5}{3}}(\vec r)A_N[\rho_{A}(\vec r);\vec r]
\label{}
\end{eqnarray}

In the present context the plausibility of this separation stems from the fact that we have defined  enhancement factors 
which are associated to a given atom, or atomic region,  and which yield satisfactory results for the non-interacting kinetic energy when the charge of the neutral species is either increased yielding a negative ion, or decreased giving a positive one.
There is some indirect  evidence, however,  that the Liu-Parr  expansion and approximation given by Eqs. (1) and (12) should work
for molecules without partition of the whole space into atomic subvolumes.\cite{3.4,Liu-et-al-1996}

But, certainly,  a division into subvolumes is required if one uses the polynomial representation
Eq. (17) for $A_N$.
Application of these ideas to molecules and clusters will be dealt with elsewhere.

\section{Conclusions.}

We have explored in the present work the possibility  of expressing the enhancement factor $A_N[\rho({\vec r}); {\vec r}]$  of the non interacting kinetic energy functional solely as a function of the one-particle density $\rho({\vec r })$, dispensing thereby with the more usual representations of this term based on gradient expansions.  This was done by adopting the representation given by Liu and Parr in terms of power series of the density.

We have analyzed the behavior of  this approximate expression for  $A_N[\rho({\vec r}); {\vec r}]$ in the case of first, second and third row atoms (with the exception of H and He, whose non-interacting kinetic energy functional is exactly given by the Weizsacker term). It is seen that the  expression for $A_N[\rho({\vec r}); {\vec r}]$ given by Eq. (10) does not comply with  the requirement of positivity.  However, when a local correction term in the form of
$\lambda\nabla^2\rho(\vec r)/\rho(\vec r)^{\nicefrac{5}{3}}$ is added to this expression,  we obtain profiles
which are in excellent agreement with those of the enhancement factors derived  from an orbital
representation. More specifically, for the second row atoms Na, Al and Ar, the locations and heights of the maxima generated by the $\lambda$-dependent approximate  $A_N[\rho({\vec r}); {\vec r}]$  (Eq. (12)) fully coincide with those obtained from an orbital representation of the  enhancement factor (Eq. (4)).  In the case of the third row atoms Fe and Ni, although the location is in perfect agreement, the maxima corresponding to the third shell fall below those of  exact ones.  An exception is the Kr atom, where the coincidence both in location and hight is quite good.  The asymptotic behavior of these $\lambda$-dependent functions near the nucleus shows a negative divergence in all cases studied. At large distances from the nucleus for the Fe and Ni atoms we observe a divergence; in all other cases the asymptote follows the trend of the exact enhancement factor. 

In addition, we have explored the possibility of introducing a $Z$-dependent
approximation to represent  the integrals $\int d\vec r\ \rho^{\nicefrac{4}{3}}(\vec r)$ and $\int d\vec r\
\rho^{\nicefrac{11}{9}}(\vec r)$ of the enhancement factor, Eq. (\ref{eq:0.7}),  through the
polynomials (\ref{eq:p1}) - (\ref{eq:p2}), respectively.  The non-interacting atomic kinetic energy density functionals generated from these new $Z$-dependent  enhancement factors (Eq.(\ref{eq:6.3})) show a behavior that very closely resembles that of the Liu-Parr functional $T_{LP97}$. From a comparison of  the MAD values for $T_{LP97}$  (0.222), $T_{LP97 + Z^3}$ (0.329), and $T_{LP97 + Z^9}$ (0.224), where all these functionals are evaluated with the Liu-Parr optimized coefficients, we see  that the functional $T_{LP97 + Z^9}$ performs as well as the Liu-Parr functional $T_{LP97}$.  However, when we re-optimize the coefficients of the $Z$-dependent functionals, we obtain the MAD values of 0.129 and 0.166 for   $T_{LP97 + Z^3}$ and $T_{LP97 + Z^9}$, respectively, thus showing a closer accord with the exact Hartree Fock values.  The behavior of the approximate enhancement factors in the case of the $Z$-dependent functionals is undistinguishable from that of the $\lambda$-corrected Liu-Parr functionals.

Concerning the extension of the $Z-\lambda$-representation of neutral atoms to positive and negative ions we see, from Tables II and III, that the  non-interacting atomic kinetic energy density functionals generated from these new $Z-\lambda$-dependent  enhancement factors perform quite well, even in the case when we use the same $Z^n$ functions  as well as the Liu-Parr coefficients which were adjusted for neutral atoms.

Summing up,  based both on the Liu-Parr power density expansion and on  the replacement of some of the integrals of this expansion by $Z$-dependent functions,  a very simple form for the non-interacting kinetic energy enhancement factor has been found.  The corresponding functionals, which bypass the usual gradient expansion representation,  lead to non-interacting kinetic energy values which closely approximate  (as measured by the MAD values) the exact ones calculated from Hartree Fock wave functions.  Moreover, the additivity  of the atomic enhancement factors, opens a possible way for  extending the present results to molecules and clusters.

\section*{Acknowledgement}

EVL would like to gratefully acknowledge SENESCYT of Ecuador for allowing him  the opportunity to participate in the Prometheus Program. VVK was supported by the US Department  of Energy TMS program, grant DE-SC0002139.

\renewcommand{\theequation}{A\arabic{equation}}
\setcounter{equation}{0}  

\section{Appendix A. Theorem 2 of Liu and Parr revisited}
For ease of understanding of general readers we have included some taken-for-granted  lines left out  in the original proof   of Theorem 2 of  Liu and Parr.

\setcounter{Teo}{1}
\begin{Teo} (S. Liu and R. G. Parr \cite{3.4}).
Given the functional
  \begin{equation}
    Q_j[\rho]=C_j[H_j]^j,\label{C4:4.10}
  \end{equation}
where  $H_j$ is a homogeneous and local functional, if it is
homogeneous of degree $m$ in coordinate scaling, it takes the
form
  \begin{equation}
    Q_j[\rho]=C_j\left[\int d\vec r \rho^{[1+(\nicefrac{m}{3j})]}(\vec r)\right]^j.\label{C4:4.11}
  \end{equation}
Further, if $Q_j[\rho]$ is homogeneous of degree $k$ in density
scaling, $j$ is determined by the relation
\begin{equation}
  j=k-\frac{m}{3}\label{C4:4.12}
\end{equation}
\end{Teo}
\begin{proof}
It is known that any strictly local functional $L[\rho]$ satisfies the identity
  \begin{equation}
    L[\rho] = -\frac{1}{3}\int d\vec r\ \vec r\cdot\nabla\rho(\vec r)\frac{\delta L[\rho]}{\delta\rho(\vec r)}
  \end{equation}
Taking the functional derivative of (\ref{C4:4.10}) with respect to $\rho$, i.e:
\begin{equation}
  \frac{\delta Q_j[\rho]}{\delta\rho} = C_j j H_j[\rho]^{j-1}\frac{\delta H_j[\rho]}{\delta\rho}\label{eq:4:14}
\end{equation}
and rewriting   Eq. (\ref{C4:4.10}), we have
\begin{align}
  Q_j[\rho]=&\ C_j(H_j[\rho])^j\nonumber\\[10pt]
  =&\ C_jH_j[\rho](H_j[\rho])^{j-1}\nonumber\\[10pt]
  =&\ C_j\left(-\frac{1}{3}\int d\vec r\ \vec r\cdot\nabla\rho(\vec r)\frac{\delta H_j[\rho]}{\delta\rho(\vec r)}\right)(H_j[\rho])^{j-1}\nonumber\\[10pt]
  =&\ -\frac{1}{3j}\int d\vec r\ \vec r\cdot\nabla\rho(\vec r)C_jj(H_j[\rho])^{j-1}\frac{\delta H_j[\rho]}{\delta\rho(\vec r)}\nonumber\\[10pt]
  =&\ -\frac{1}{3j}\int d\vec r\ \vec r\cdot\nabla\rho(\vec r)\frac{\delta Q_j[\rho]}{\delta\rho(\vec r)}.\label{C4:13}
\end{align}
Because $Q_j$ is homogeneous of degree $m$ in coordinate scaling it follows that
  \begin{equation}
    -\int d\vec r\ \rho(\vec r)\vec r\cdot\nabla\frac{\delta Q_j[\rho]}{\delta\rho(\vec r)} = m Q_j[\rho].\label{C4:14}
  \end{equation}
Thus, if we integrate this equation by parts, we obtain
 \begin{align}
    -\int d\vec r\ \rho(\vec r)\vec r\cdot\nabla\frac{\delta Q_j[\rho]}{\delta\rho(\vec r)} =& \int d\vec r\ (\vec r \cdot \nabla\rho(\vec r)\nonumber\\
     &\ + 3\rho(\vec r))\frac{\delta Q_j[\rho]}{\delta\rho(\vec r)}\label{C4:17}
  \end{align}
and by replacing  Eq. (\ref{C4:13}) into   (\ref{C4:17}), it is found
  \begin{equation}
    \int d\vec r\ \rho(\vec r)\frac{\delta Q_j[\rho]}{\delta\rho(\vec r)} = \frac{m+3j}{3} Q_j[\rho].\label{C4:15}
  \end{equation}
This shows that $Q_j[\rho]$ is homogeneous of degree $(m+3j)/3$ in coordinate scaling. On the other hand, $H_j[\rho]$ is homogeneous, i.e.
\begin{equation}
  H_j[\rho] = \int d\vec r\ f_j(\rho(\vec r)),\label{C4:19}
\end{equation}
so that if we replace Eq. (\ref{C4:19}) into Eqs. (\ref{C4:4.10}) and (\ref{eq:4:14}) and these in turn into Eq.  (\ref{C4:15}) we have
\begin{align}
  \int d \vec r\ \rho(\vec r) j C_j (H_j[\rho])^{j-1}\frac{\delta H_j[\rho]}{\delta\rho(\vec r)} = & \frac{m + 3j}{3}C_j(H_j[\rho])^j\nonumber\\[10pt]
  \int d \vec r\ \rho(\vec r)\frac{\delta H_j[\rho]}{\delta\rho(\vec r)} = & \left(1+ \frac{m}{3j}\right)(H_j[\rho])\nonumber\\[10pt]
  \int d \vec r\ \rho(\vec r)\frac{d f_j(\rho)}{d\rho(\vec r)} = & \left(1+ \frac{m}{3j}\right)\int d\vec r\ f_j(\rho)\nonumber\\[10pt]
  \int d \vec r\ \rho(\vec r)\frac{d f_j(\rho)}{d\rho(\vec r)} = & \int d\vec r\ \left(1+ \frac{m}{3j}\right)f_j(\rho).
\end{align}
If the two integrals are equal it follows:
\begin{equation}
  \rho(\vec r)\frac{d f_j(\rho)}{d\rho(\vec r)} = \left(1+ \frac{m}{3j}\right)f_j(\rho).
\end{equation}
Therefore, we have to solve a simple differential equation
\begin{align}
  \int\frac{df_j(\rho)}{f_j(\rho)} =&\ \int\left(1+ \frac{m}{3j}\right)\frac{d\rho(\vec
r)}{\rho(\vec r)}\nonumber\\[10pt]
  \ln f_j(\rho) =&\ \left(1+ \frac{m}{3j}\right)\ln\rho(\vec r) + C_j\nonumber\\[10pt]
  f_j(\rho) = &\ C_j\rho^{[1 + \nicefrac{m}{3j}]}(\vec r),
\end{align}
where $C_j$ is a constant of integration. This leads to:
\begin{equation}
  H_j[\rho] = C_j\int d\vec r\ \rho^{[1 + \nicefrac{m}{3j}]}(\vec r)
\end{equation}
and
\begin{equation}
  Q_j[\rho] = C_j\left[\int d\vec r\ \rho^{[1 + \nicefrac{m}{3j}]}(\vec r)\right]^j.
\end{equation}
Finally, we see from Eq. (\ref{C4:15}) that $k$ is $(m + 3j)/3$, thus
  \begin{equation}
    j = k - \frac{m}{3}.
  \end{equation}
\end{proof}
\bibliographystyle{apsrev4-1}
\addcontentsline{toc}{chapter}{Referencias}
\def\bibname{Referencias}
\bibliography{base1}

\begin{thebibliography}{56}%
\makeatletter
\providecommand \@ifxundefined [1]{%
 \@ifx{#1\undefined}
}%
\providecommand \@ifnum [1]{%
 \ifnum #1\expandafter \@firstoftwo
 \else \expandafter \@secondoftwo
 \fi
}%
\providecommand \@ifx [1]{%
 \ifx #1\expandafter \@firstoftwo
 \else \expandafter \@secondoftwo
 \fi
}%
\providecommand \natexlab [1]{#1}%
\providecommand \enquote  [1]{``#1''}%
\providecommand \bibnamefont  [1]{#1}%
\providecommand \bibfnamefont [1]{#1}%
\providecommand \citenamefont [1]{#1}%
\providecommand \href@noop [0]{\@secondoftwo}%
\providecommand \href [0]{\begingroup \@sanitize@url \@href}%
\providecommand \@href[1]{\@@startlink{#1}\@@href}%
\providecommand \@@href[1]{\endgroup#1\@@endlink}%
\providecommand \@sanitize@url [0]{\catcode `\\12\catcode `\$12\catcode
  `\&12\catcode `\#12\catcode `\^12\catcode `\_12\catcode `\%12\relax}%
\providecommand \@@startlink[1]{}%
\providecommand \@@endlink[0]{}%
\providecommand \url  [0]{\begingroup\@sanitize@url \@url }%
\providecommand \@url [1]{\endgroup\@href {#1}{\urlprefix }}%
\providecommand \urlprefix  [0]{URL }%
\providecommand \Eprint [0]{\href }%
\providecommand \doibase [0]{http://dx.doi.org/}%
\providecommand \selectlanguage [0]{\@gobble}%
\providecommand \bibinfo  [0]{\@secondoftwo}%
\providecommand \bibfield  [0]{\@secondoftwo}%
\providecommand \translation [1]{[#1]}%
\providecommand \BibitemOpen [0]{}%
\providecommand \bibitemStop [0]{}%
\providecommand \bibitemNoStop [0]{.\EOS\space}%
\providecommand \EOS [0]{\spacefactor3000\relax}%
\providecommand \BibitemShut  [1]{\csname bibitem#1\endcsname}%
\let\auto@bib@innerbib\@empty
\bibitem [{\citenamefont {Liu}\ and\ \citenamefont {Parr}(1997)}]{3.4}%
  \BibitemOpen
  \bibfield  {author} {\bibinfo {author} {\bibfnamefont {S.}~\bibnamefont
  {Liu}}\ and\ \bibinfo {author} {\bibfnamefont {R.~G.}\ \bibnamefont {Parr}},\
  }\href {\doibase 10.1103/PhysRevA.55.1792} {\bibfield  {journal} {\bibinfo
  {journal} {Phys. Rev. A}\ }\textbf {\bibinfo {volume} {55}},\ \bibinfo
  {pages} {1792} (\bibinfo {year} {1997})}\BibitemShut {NoStop}%
\bibitem [{\citenamefont {Karasiev}\ \emph {et~al.}(2000)\citenamefont
  {Karasiev}, \citenamefont {Lude\~na},\ and\ \citenamefont {Artemyev}}]{1.5}%
  \BibitemOpen
  \bibfield  {author} {\bibinfo {author} {\bibfnamefont {V.~V.}\ \bibnamefont
  {Karasiev}}, \bibinfo {author} {\bibfnamefont {E.~V.}\ \bibnamefont
  {Lude\~na}}, \ and\ \bibinfo {author} {\bibfnamefont {A.~N.}\ \bibnamefont
  {Artemyev}},\ }\href {\doibase 10.1103/PhysRevA.62.062510} {\bibfield
  {journal} {\bibinfo  {journal} {Phys. Rev. A}\ }\textbf {\bibinfo {volume}
  {62}},\ \bibinfo {pages} {062510} (\bibinfo {year} {2000})}\BibitemShut
  {NoStop}%
\bibitem [{\citenamefont {G\'al}\ and\ \citenamefont {Nagy}(2000)}]{1.10}%
  \BibitemOpen
  \bibfield  {author} {\bibinfo {author} {\bibfnamefont {T.}~\bibnamefont
  {G\'al}}\ and\ \bibinfo {author} {\bibfnamefont {A.}~\bibnamefont {Nagy}},\
  }\href {\doibase 10.1016/S0166-1280(99)00425-X} {\bibfield  {journal}
  {\bibinfo  {journal} {Journal of Molecular Structure: THEOCHEM}\ }\textbf
  {\bibinfo {volume} {501 - 502}},\ \bibinfo {pages} {167 } (\bibinfo {year}
  {2000})}\BibitemShut {NoStop}%
\bibitem [{\citenamefont {Karasiev}\ \emph {et~al.}(2012)\citenamefont
  {Karasiev}, \citenamefont {Sjostrom},\ and\ \citenamefont {Trickey}}]{1.12}%
  \BibitemOpen
  \bibfield  {author} {\bibinfo {author} {\bibfnamefont {V.~V.}\ \bibnamefont
  {Karasiev}}, \bibinfo {author} {\bibfnamefont {T.}~\bibnamefont {Sjostrom}},
  \ and\ \bibinfo {author} {\bibfnamefont {S.~B.}\ \bibnamefont {Trickey}},\
  }\href {\doibase 10.1103/PhysRevB.86.115101} {\bibfield  {journal} {\bibinfo
  {journal} {Phys. Rev. B}\ }\textbf {\bibinfo {volume} {86}},\ \bibinfo
  {pages} {115101} (\bibinfo {year} {2012})}\BibitemShut {NoStop}%
\bibitem [{\citenamefont {Karasiev}\ and\ \citenamefont
  {Trickey}(2012)}]{1.11}%
  \BibitemOpen
  \bibfield  {author} {\bibinfo {author} {\bibfnamefont {V.~V.}\ \bibnamefont
  {Karasiev}}\ and\ \bibinfo {author} {\bibfnamefont {S.~B.}\ \bibnamefont
  {Trickey}},\ }\href@noop {} {\bibfield  {journal} {\bibinfo  {journal}
  {Computer Physics Communications}\ }\textbf {\bibinfo {volume} {183}},\
  \bibinfo {pages} {2519} (\bibinfo {year} {2012})}\BibitemShut {NoStop}%
\bibitem [{\citenamefont {Tran}\ and\ \citenamefont {Wesolowski}(2013)}]{Apro}%
  \BibitemOpen
  \bibfield  {author} {\bibinfo {author} {\bibfnamefont {F.}~\bibnamefont
  {Tran}}\ and\ \bibinfo {author} {\bibfnamefont {T.~A.}\ \bibnamefont
  {Wesolowski}},\ }\enquote {\bibinfo {title} {Semilocal approximations for the
  kinetic energy},}\ \ (\bibinfo  {publisher} {World Scientific},\ \bibinfo
  {year} {2013})\ pp.\ \bibinfo {pages} {429--442}\BibitemShut {NoStop}%
\bibitem [{\citenamefont {Iyengar}\ \emph {et~al.}(2001)\citenamefont
  {Iyengar}, \citenamefont {Ernzerhof}, \citenamefont {Maximoff},\ and\
  \citenamefont {Scuseria}}]{Iyengar-2001}%
  \BibitemOpen
  \bibfield  {author} {\bibinfo {author} {\bibfnamefont {S.~S.}\ \bibnamefont
  {Iyengar}}, \bibinfo {author} {\bibfnamefont {M.}~\bibnamefont {Ernzerhof}},
  \bibinfo {author} {\bibfnamefont {S.~N.}\ \bibnamefont {Maximoff}}, \ and\
  \bibinfo {author} {\bibfnamefont {G.~E.}\ \bibnamefont {Scuseria}},\ }\href
  {\doibase 10.1103/PhysRevA.63.052508} {\bibfield  {journal} {\bibinfo
  {journal} {Phys. Rev. A}\ }\textbf {\bibinfo {volume} {63}},\ \bibinfo
  {pages} {052508} (\bibinfo {year} {2001})}\BibitemShut {NoStop}%
\bibitem [{\citenamefont {Karasiev}\ \emph
  {et~al.}(2014{\natexlab{a}})\citenamefont {Karasiev}, \citenamefont
  {Chakraborty},\ and\ \citenamefont {Trickey}}]{Karasiev-et-al-2014-Bach}%
  \BibitemOpen
  \bibfield  {author} {\bibinfo {author} {\bibfnamefont {V.}~\bibnamefont
  {Karasiev}}, \bibinfo {author} {\bibfnamefont {D.}~\bibnamefont
  {Chakraborty}}, \ and\ \bibinfo {author} {\bibfnamefont {S.}~\bibnamefont
  {Trickey}},\ }in\ \href {\doibase 10.1007/978-3-319-06379-9_6} {\emph
  {\bibinfo {booktitle} {Many-Electron Approaches in Physics, Chemistry and
  Mathematics}}},\ \bibinfo {series and number} {Mathematical Physics
  Studies},\ \bibinfo {editor} {edited by\ \bibinfo {editor} {\bibfnamefont
  {V.}~\bibnamefont {Bach}}\ and\ \bibinfo {editor} {\bibfnamefont
  {L.}~\bibnamefont {Delle~Site}}}\ (\bibinfo  {publisher} {Springer
  International Publishing},\ \bibinfo {year} {2014})\ pp.\ \bibinfo {pages}
  {113--134}\BibitemShut {NoStop}%
\bibitem [{\citenamefont {Karasiev}\ and\ \citenamefont
  {Trickey}(2015)}]{Karasiev-et-al-2015}%
  \BibitemOpen
  \bibfield  {author} {\bibinfo {author} {\bibfnamefont {V.~V.}\ \bibnamefont
  {Karasiev}}\ and\ \bibinfo {author} {\bibfnamefont {S.~B.}\ \bibnamefont
  {Trickey}},\ }in\ \href {\doibase
  http://dx.doi.org/10.1016/bs.aiq.2015.02.004} {\emph {\bibinfo {booktitle}
  {Concepts of Mathematical Physics in Chemistry: A Tribute to Frank E. Harris
  - Part A}}},\ \bibinfo {series} {Advances in Quantum Chemistry},
  Vol.~\bibinfo {volume} {71},\ \bibinfo {editor} {edited by\ \bibinfo {editor}
  {\bibfnamefont {J.~R.}\ \bibnamefont {Sabin}}\ and\ \bibinfo {editor}
  {\bibfnamefont {R.}~\bibnamefont {Cabrera-Trujillo}}}\ (\bibinfo  {publisher}
  {Academic Press},\ \bibinfo {year} {2015})\ pp.\ \bibinfo {pages} {221 --
  245}\BibitemShut {NoStop}%
\bibitem [{\citenamefont {Wang}\ and\ \citenamefont
  {Carter}(2002)}]{Wang-Carter-2000}%
  \BibitemOpen
  \bibfield  {author} {\bibinfo {author} {\bibfnamefont {Y.~A.}\ \bibnamefont
  {Wang}}\ and\ \bibinfo {author} {\bibfnamefont {E.~A.}\ \bibnamefont
  {Carter}},\ }in\ \href@noop {} {\emph {\bibinfo {booktitle} {Theoretical
  methods in condensed phase chemistry}}}\ (\bibinfo  {publisher} {Springer
  Netherlands},\ \bibinfo {year} {2002})\ pp.\ \bibinfo {pages}
  {117--184}\BibitemShut {NoStop}%
\bibitem [{\citenamefont {Kohn}\ and\ \citenamefont {Sham}(1965)}]{1.3}%
  \BibitemOpen
  \bibfield  {author} {\bibinfo {author} {\bibfnamefont {W.}~\bibnamefont
  {Kohn}}\ and\ \bibinfo {author} {\bibfnamefont {L.~J.}\ \bibnamefont
  {Sham}},\ }\href {\doibase 10.1103/PhysRev.140.A1133} {\bibfield  {journal}
  {\bibinfo  {journal} {Phys. Rev.}\ }\textbf {\bibinfo {volume} {140}},\
  \bibinfo {pages} {A1133} (\bibinfo {year} {1965})}\BibitemShut {NoStop}%
\bibitem [{\citenamefont {Thomas}(1927)}]{1.1}%
  \BibitemOpen
  \bibfield  {author} {\bibinfo {author} {\bibfnamefont {L.~H.}\ \bibnamefont
  {Thomas}},\ }\href {\doibase 10.1017/S0305004100011683} {\bibfield  {journal}
  {\bibinfo  {journal} {Mathematical Proceedings of the Cambridge Philosophical
  Society}\ }\textbf {\bibinfo {volume} {23}},\ \bibinfo {pages} {542}
  (\bibinfo {year} {1927})}\BibitemShut {NoStop}%
\bibitem [{\citenamefont {Fiorini}\ and\ \citenamefont
  {Gallavotti}(2011)}]{1.2}%
  \BibitemOpen
  \bibfield  {author} {\bibinfo {author} {\bibfnamefont {E.}~\bibnamefont
  {Fiorini}}\ and\ \bibinfo {author} {\bibfnamefont {G.}~\bibnamefont
  {Gallavotti}},\ }\href {http://dx.doi.org/10.1007/s12210-011-0141-5}
  {\bibfield  {journal} {\bibinfo  {journal} {Rendiconti Lincei}\ }\textbf
  {\bibinfo {volume} {22}},\ \bibinfo {pages} {277} (\bibinfo {year}
  {2011})}\BibitemShut {NoStop}%
\bibitem [{\citenamefont {Lude\~na}\ and\ \citenamefont
  {Karasiev}(2002)}]{2.1}%
  \BibitemOpen
  \bibfield  {author} {\bibinfo {author} {\bibfnamefont {E.~V.}\ \bibnamefont
  {Lude\~na}}\ and\ \bibinfo {author} {\bibfnamefont {V.}~\bibnamefont
  {Karasiev}},\ }in\ \href {http://www.wspc.com.sg/books/chemistry/4910.html}
  {\emph {\bibinfo {booktitle} {A celebration of the contributions of Robert
  Parr}}},\ Vol.~\bibinfo {volume} {1},\ \bibinfo {editor} {edited by\ \bibinfo
  {editor} {\bibfnamefont {K.~D.}\ \bibnamefont {Sen}}}\ (\bibinfo  {publisher}
  {World Scientific},\ \bibinfo {year} {2002})\ pp.\ \bibinfo {pages}
  {1--55}\BibitemShut {NoStop}%
\bibitem [{\citenamefont {Lindmaa}\ \emph {et~al.}(2014)\citenamefont
  {Lindmaa}, \citenamefont {Mattsson},\ and\ \citenamefont
  {Armiento}}]{Lindmaa-2014}%
  \BibitemOpen
  \bibfield  {author} {\bibinfo {author} {\bibfnamefont {A.}~\bibnamefont
  {Lindmaa}}, \bibinfo {author} {\bibfnamefont {A.~E.}\ \bibnamefont
  {Mattsson}}, \ and\ \bibinfo {author} {\bibfnamefont {R.}~\bibnamefont
  {Armiento}},\ }\href {\doibase 10.1103/PhysRevB.90.075139} {\bibfield
  {journal} {\bibinfo  {journal} {Phys. Rev. B}\ }\textbf {\bibinfo {volume}
  {90}},\ \bibinfo {pages} {075139} (\bibinfo {year} {2014})}\BibitemShut
  {NoStop}%
\bibitem [{\citenamefont {Higuchi}\ and\ \citenamefont
  {Higuchi}()}]{Higuchi-2014}%
  \BibitemOpen
  \bibfield  {author} {\bibinfo {author} {\bibfnamefont {K.}~\bibnamefont
  {Higuchi}}\ and\ \bibinfo {author} {\bibfnamefont {M.}~\bibnamefont
  {Higuchi}},\ }\enquote {\bibinfo {title} {A proposal of the approximate
  kinetic energy functional of the pair density functional theory},}\ in\ \href
  {\doibase 10.7566/JPSCP.3.017009} {\emph {\bibinfo {booktitle} {Proceedings
  of the International Conference on Strongly Correlated Electron Systems
  (SCES2013)}}},\ Chap.\ \bibinfo {chapter} {273},\ \Eprint
  {http://arxiv.org/abs/http://journals.jps.jp/doi/pdf/10.7566/JPSCP.3.017009}
  {http://journals.jps.jp/doi/pdf/10.7566/JPSCP.3.017009} \BibitemShut
  {NoStop}%
\bibitem [{\citenamefont {Laricchia}\ \emph {et~al.}(2014)\citenamefont
  {Laricchia}, \citenamefont {Constantin}, \citenamefont {Fabiano},\ and\
  \citenamefont {Della~Sala}}]{2014-Laricchia}%
  \BibitemOpen
  \bibfield  {author} {\bibinfo {author} {\bibfnamefont {S.}~\bibnamefont
  {Laricchia}}, \bibinfo {author} {\bibfnamefont {L.~A.}\ \bibnamefont
  {Constantin}}, \bibinfo {author} {\bibfnamefont {E.}~\bibnamefont {Fabiano}},
  \ and\ \bibinfo {author} {\bibfnamefont {F.}~\bibnamefont {Della~Sala}},\
  }\href {\doibase 10.1021/ct400836s} {\bibfield  {journal} {\bibinfo
  {journal} {Journal of Chemical Theory and Computation}\ }\textbf {\bibinfo
  {volume} {10}},\ \bibinfo {pages} {164} (\bibinfo {year} {2014})},\ \Eprint
  {http://arxiv.org/abs/http://dx.doi.org/10.1021/ct400836s}
  {http://dx.doi.org/10.1021/ct400836s} \BibitemShut {NoStop}%
\bibitem [{\citenamefont {Xia}\ and\ \citenamefont
  {Carter}(2015)}]{Xia-carter-2015}%
  \BibitemOpen
  \bibfield  {author} {\bibinfo {author} {\bibfnamefont {J.}~\bibnamefont
  {Xia}}\ and\ \bibinfo {author} {\bibfnamefont {E.~A.}\ \bibnamefont
  {Carter}},\ }\href {\doibase 10.1103/PhysRevB.91.045124} {\bibfield
  {journal} {\bibinfo  {journal} {Phys. Rev. B}\ }\textbf {\bibinfo {volume}
  {91}},\ \bibinfo {pages} {045124} (\bibinfo {year} {2015})}\BibitemShut
  {NoStop}%
\bibitem [{\citenamefont {{Sergeev}}\ \emph {et~al.}(2014)\citenamefont
  {{Sergeev}}, \citenamefont {{Alharbi}}, \citenamefont {{Jovanovic}},\ and\
  \citenamefont {{Kais}}}]{sergeev}%
  \BibitemOpen
  \bibfield  {author} {\bibinfo {author} {\bibfnamefont {A.}~\bibnamefont
  {{Sergeev}}}, \bibinfo {author} {\bibfnamefont {F.~H.}\ \bibnamefont
  {{Alharbi}}}, \bibinfo {author} {\bibfnamefont {R.}~\bibnamefont
  {{Jovanovic}}}, \ and\ \bibinfo {author} {\bibfnamefont {S.}~\bibnamefont
  {{Kais}}},\ }\href@noop {} {\bibfield  {journal} {\bibinfo  {journal} {ArXiv
  e-prints}\ } (\bibinfo {year} {2014})},\ \Eprint
  {http://arxiv.org/abs/1411.0804} {arXiv:1411.0804 [quant-ph]} \BibitemShut
  {NoStop}%
\bibitem [{\citenamefont {Della~Sala}\ \emph {et~al.}(2015)\citenamefont
  {Della~Sala}, \citenamefont {Fabiano},\ and\ \citenamefont
  {Constantin}}]{DellaSala-2015}%
  \BibitemOpen
  \bibfield  {author} {\bibinfo {author} {\bibfnamefont {F.}~\bibnamefont
  {Della~Sala}}, \bibinfo {author} {\bibfnamefont {E.}~\bibnamefont {Fabiano}},
  \ and\ \bibinfo {author} {\bibfnamefont {L.~A.}\ \bibnamefont {Constantin}},\
  }\href {\doibase 10.1103/PhysRevB.91.035126} {\bibfield  {journal} {\bibinfo
  {journal} {Phys. Rev. B}\ }\textbf {\bibinfo {volume} {91}},\ \bibinfo
  {pages} {035126} (\bibinfo {year} {2015})}\BibitemShut {NoStop}%
\bibitem [{\citenamefont {Espinosa~Leal}\ \emph {et~al.}(2015)\citenamefont
  {Espinosa~Leal}, \citenamefont {Karpenko}, \citenamefont {Caro},\ and\
  \citenamefont {Lopez-Acevedo}}]{Espinosa-Leal-2015}%
  \BibitemOpen
  \bibfield  {author} {\bibinfo {author} {\bibfnamefont {L.~A.}\ \bibnamefont
  {Espinosa~Leal}}, \bibinfo {author} {\bibfnamefont {A.}~\bibnamefont
  {Karpenko}}, \bibinfo {author} {\bibfnamefont {M.~A.}\ \bibnamefont {Caro}},
  \ and\ \bibinfo {author} {\bibfnamefont {O.}~\bibnamefont {Lopez-Acevedo}},\
  }\href {\doibase 10.1039/C5CP01211B} {\bibfield  {journal} {\bibinfo
  {journal} {Phys. Chem. Chem. Phys.}\ ,\ } (\bibinfo {year}
  {2015})}\BibitemShut {NoStop}%
\bibitem [{\citenamefont {Finzel}(2015)}]{Finzel-2015}%
  \BibitemOpen
  \bibfield  {author} {\bibinfo {author} {\bibfnamefont {K.}~\bibnamefont
  {Finzel}},\ }\href {\doibase 10.1007/s00214-015-1711-x} {\bibfield  {journal}
  {\bibinfo  {journal} {Theoretical Chemistry Accounts}\ }\textbf {\bibinfo
  {volume} {134}},\ \bibinfo {eid} {106} (\bibinfo {year} {2015}),\
  10.1007/s00214-015-1711-x}\BibitemShut {NoStop}%
\bibitem [{\citenamefont {Li}\ \emph {et~al.}(2015)\citenamefont {Li},
  \citenamefont {Snyder}, \citenamefont {Pelaschier}, \citenamefont {Huang},
  \citenamefont {Niranjan}, \citenamefont {Duncan}, \citenamefont {Rupp},
  \citenamefont {Müller},\ and\ \citenamefont {Burke}}]{Li-et-al-2015}%
  \BibitemOpen
  \bibfield  {author} {\bibinfo {author} {\bibfnamefont {L.}~\bibnamefont
  {Li}}, \bibinfo {author} {\bibfnamefont {J.~C.}\ \bibnamefont {Snyder}},
  \bibinfo {author} {\bibfnamefont {I.~M.}\ \bibnamefont {Pelaschier}},
  \bibinfo {author} {\bibfnamefont {J.}~\bibnamefont {Huang}}, \bibinfo
  {author} {\bibfnamefont {U.-N.}\ \bibnamefont {Niranjan}}, \bibinfo {author}
  {\bibfnamefont {P.}~\bibnamefont {Duncan}}, \bibinfo {author} {\bibfnamefont
  {M.}~\bibnamefont {Rupp}}, \bibinfo {author} {\bibfnamefont {K.-R.}\
  \bibnamefont {Müller}}, \ and\ \bibinfo {author} {\bibfnamefont
  {K.}~\bibnamefont {Burke}},\ }\href {\doibase 10.1002/qua.25040} {\bibfield
  {journal} {\bibinfo  {journal} {International Journal of Quantum Chemistry}\
  ,\ \bibinfo {pages} {n/a}} (\bibinfo {year} {2015})}\BibitemShut {NoStop}%
\bibitem [{\citenamefont {Ludeña}\ \emph {et~al.}(1996)\citenamefont
  {Ludeña}, \citenamefont {López-Boada},\ and\ \citenamefont
  {Pino}}]{EVL-CanJChem-1996}%
  \BibitemOpen
  \bibfield  {author} {\bibinfo {author} {\bibfnamefont {E.}~\bibnamefont
  {Ludeña}}, \bibinfo {author} {\bibfnamefont {R.}~\bibnamefont
  {López-Boada}}, \ and\ \bibinfo {author} {\bibfnamefont {R.}~\bibnamefont
  {Pino}},\ }\href {\doibase 10.1139/v96-123} {\bibfield  {journal} {\bibinfo
  {journal} {Canadian Journal of Chemistry}\ }\textbf {\bibinfo {volume}
  {74}},\ \bibinfo {pages} {1097} (\bibinfo {year} {1996})},\ \Eprint
  {http://arxiv.org/abs/http://dx.doi.org/10.1139/v96-123}
  {http://dx.doi.org/10.1139/v96-123} \BibitemShut {NoStop}%
\bibitem [{\citenamefont {Lude\~na}\ \emph {et~al.}(1999)\citenamefont
  {Lude\~na}, \citenamefont {Karasiev}, \citenamefont {L\'opez-Boada},
  \citenamefont {Valderrama},\ and\ \citenamefont {Maldonado}}]{1.9}%
  \BibitemOpen
  \bibfield  {author} {\bibinfo {author} {\bibfnamefont {E.~V.}\ \bibnamefont
  {Lude\~na}}, \bibinfo {author} {\bibfnamefont {V.}~\bibnamefont {Karasiev}},
  \bibinfo {author} {\bibfnamefont {R.}~\bibnamefont {L\'opez-Boada}}, \bibinfo
  {author} {\bibfnamefont {E.}~\bibnamefont {Valderrama}}, \ and\ \bibinfo
  {author} {\bibfnamefont {J.}~\bibnamefont {Maldonado}},\ }\href {\doibase
  10.1002/(SICI)1096-987X(19990115)20:1<155::AID-JCC14>3.0.CO;2-2} {\bibfield
  {journal} {\bibinfo  {journal} {Journal of Computational Chemistry}\ }\textbf
  {\bibinfo {volume} {20}},\ \bibinfo {pages} {155} (\bibinfo {year}
  {1999})}\BibitemShut {NoStop}%
\bibitem [{\citenamefont {Lude\~na}\ \emph {et~al.}(2003)\citenamefont
  {Lude\~na}, \citenamefont {Karasiev},\ and\ \citenamefont
  {Echevarr\'ia}}]{1.4}%
  \BibitemOpen
  \bibfield  {author} {\bibinfo {author} {\bibfnamefont {E.~V.}\ \bibnamefont
  {Lude\~na}}, \bibinfo {author} {\bibfnamefont {V.~V.}\ \bibnamefont
  {Karasiev}}, \ and\ \bibinfo {author} {\bibfnamefont {L.}~\bibnamefont
  {Echevarr\'ia}},\ }\href {\doibase 10.1002/qua.10401} {\bibfield  {journal}
  {\bibinfo  {journal} {International Journal of Quantum Chemistry}\ }\textbf
  {\bibinfo {volume} {91}},\ \bibinfo {pages} {94} (\bibinfo {year}
  {2003})}\BibitemShut {NoStop}%
\bibitem [{\citenamefont {von Weizs\"acker}(1935)}]{2.2}%
  \BibitemOpen
  \bibfield  {author} {\bibinfo {author} {\bibfnamefont {C.~F.}\ \bibnamefont
  {von Weizs\"acker}},\ }\href@noop {} {\bibfield  {journal} {\bibinfo
  {journal} {Z Phys}\ }\textbf {\bibinfo {volume} {96}},\ \bibinfo {pages}
  {431} (\bibinfo {year} {1935})}\BibitemShut {NoStop}%
\bibitem [{\citenamefont {Kristyan}(2013)}]{Kristian-2013}%
  \BibitemOpen
  \bibfield  {author} {\bibinfo {author} {\bibfnamefont {S.}~\bibnamefont
  {Kristyan}},\ }\href {\doibase 10.1186/2251-7235-7-61} {\bibfield  {journal}
  {\bibinfo  {journal} {Journal of Theoretical and Applied Physics}\ }\textbf
  {\bibinfo {volume} {7}},\ \bibinfo {eid} {61} (\bibinfo {year} {2013}),\
  10.1186/2251-7235-7-61}\BibitemShut {NoStop}%
\bibitem [{\citenamefont {Sears}\ \emph {et~al.}(1980)\citenamefont {Sears},
  \citenamefont {Parr},\ and\ \citenamefont {Dinur}}]{Sears-Parr-Dinur}%
  \BibitemOpen
  \bibfield  {author} {\bibinfo {author} {\bibfnamefont {S.~B.}\ \bibnamefont
  {Sears}}, \bibinfo {author} {\bibfnamefont {R.~G.}\ \bibnamefont {Parr}}, \
  and\ \bibinfo {author} {\bibfnamefont {U.}~\bibnamefont {Dinur}},\ }\href
  {\doibase 10.1002/ijch.198000018} {\bibfield  {journal} {\bibinfo  {journal}
  {Israel Journal of Chemistry}\ }\textbf {\bibinfo {volume} {19}},\ \bibinfo
  {pages} {165} (\bibinfo {year} {1980})}\BibitemShut {NoStop}%
\bibitem [{\citenamefont {Lude\~na}(1982)}]{Ludena-1982}%
  \BibitemOpen
  \bibfield  {author} {\bibinfo {author} {\bibfnamefont {E.~V.}\ \bibnamefont
  {Lude\~na}},\ }\href {\doibase http://dx.doi.org/10.1063/1.443358} {\bibfield
   {journal} {\bibinfo  {journal} {The Journal of Chemical Physics}\ }\textbf
  {\bibinfo {volume} {76}},\ \bibinfo {pages} {3157} (\bibinfo {year}
  {1982})}\BibitemShut {NoStop}%
\bibitem [{\citenamefont {Lude\~na}\ \emph {et~al.}(2004)\citenamefont
  {Lude\~na}, \citenamefont {Ugalde}, \citenamefont {Lopez}, \citenamefont
  {Fern\'andez-Rico},\ and\ \citenamefont {Ram\'irez}}]{Ludena-et-al-JCP-2004}%
  \BibitemOpen
  \bibfield  {author} {\bibinfo {author} {\bibfnamefont {E.~V.}\ \bibnamefont
  {Lude\~na}}, \bibinfo {author} {\bibfnamefont {J.~M.}\ \bibnamefont
  {Ugalde}}, \bibinfo {author} {\bibfnamefont {X.}~\bibnamefont {Lopez}},
  \bibinfo {author} {\bibfnamefont {J.}~\bibnamefont {Fern\'andez-Rico}}, \
  and\ \bibinfo {author} {\bibfnamefont {G.}~\bibnamefont {Ram\'irez}},\ }\href
  {\doibase http://dx.doi.org/10.1063/1.1630024} {\bibfield  {journal}
  {\bibinfo  {journal} {The Journal of Chemical Physics}\ }\textbf {\bibinfo
  {volume} {120}},\ \bibinfo {pages} {540} (\bibinfo {year}
  {2004})}\BibitemShut {NoStop}%
\bibitem [{\citenamefont {Tal}\ and\ \citenamefont {Bader}(1978)}]{Tal-Bader}%
  \BibitemOpen
  \bibfield  {author} {\bibinfo {author} {\bibfnamefont {Y.}~\bibnamefont
  {Tal}}\ and\ \bibinfo {author} {\bibfnamefont {R.~F.~W.}\ \bibnamefont
  {Bader}},\ }\href {\doibase 10.1002/qua.560140813} {\bibfield  {journal}
  {\bibinfo  {journal} {International Journal of Quantum Chemistry}\ }\textbf
  {\bibinfo {volume} {14}},\ \bibinfo {pages} {153} (\bibinfo {year}
  {1978})}\BibitemShut {NoStop}%
\bibitem [{\citenamefont {Becke}\ and\ \citenamefont
  {Edgecombe}(1990)}]{Becke-Edgecombe-1990}%
  \BibitemOpen
  \bibfield  {author} {\bibinfo {author} {\bibfnamefont {A.~D.}\ \bibnamefont
  {Becke}}\ and\ \bibinfo {author} {\bibfnamefont {K.~E.}\ \bibnamefont
  {Edgecombe}},\ }\href {\doibase http://dx.doi.org/10.1063/1.458517}
  {\bibfield  {journal} {\bibinfo  {journal} {The Journal of Chemical Physics}\
  }\textbf {\bibinfo {volume} {92}},\ \bibinfo {pages} {5397} (\bibinfo {year}
  {1990})}\BibitemShut {NoStop}%
\bibitem [{\citenamefont {Savin}\ \emph {et~al.}(1997)\citenamefont {Savin},
  \citenamefont {Nesper}, \citenamefont {Wengert},\ and\ \citenamefont
  {Fässler}}]{Savin-1997}%
  \BibitemOpen
  \bibfield  {author} {\bibinfo {author} {\bibfnamefont {A.}~\bibnamefont
  {Savin}}, \bibinfo {author} {\bibfnamefont {R.}~\bibnamefont {Nesper}},
  \bibinfo {author} {\bibfnamefont {S.}~\bibnamefont {Wengert}}, \ and\
  \bibinfo {author} {\bibfnamefont {T.~F.}\ \bibnamefont {Fässler}},\ }\href
  {\doibase 10.1002/anie.199718081} {\bibfield  {journal} {\bibinfo  {journal}
  {Angewandte Chemie International Edition in English}\ }\textbf {\bibinfo
  {volume} {36}},\ \bibinfo {pages} {1808} (\bibinfo {year}
  {1997})}\BibitemShut {NoStop}%
\bibitem [{\citenamefont {Savin}(2005{\natexlab{a}})}]{Savin-2005-ELF}%
  \BibitemOpen
  \bibfield  {author} {\bibinfo {author} {\bibfnamefont {A.}~\bibnamefont
  {Savin}},\ }\href {\doibase 10.1007/BF02708351} {\bibfield  {journal}
  {\bibinfo  {journal} {Journal of Chemical Sciences}\ }\textbf {\bibinfo
  {volume} {117}},\ \bibinfo {pages} {473} (\bibinfo {year}
  {2005}{\natexlab{a}})}\BibitemShut {NoStop}%
\bibitem [{\citenamefont {Savin}(2005{\natexlab{b}})}]{Savin-2005-Theochem}%
  \BibitemOpen
  \bibfield  {author} {\bibinfo {author} {\bibfnamefont {A.}~\bibnamefont
  {Savin}},\ }\href {\doibase http://dx.doi.org/10.1016/j.theochem.2005.02.034}
  {\bibfield  {journal} {\bibinfo  {journal} {Journal of Molecular Structure:
  \{THEOCHEM\}}\ }\textbf {\bibinfo {volume} {727}},\ \bibinfo {pages} {127 }
  (\bibinfo {year} {2005}{\natexlab{b}})}\BibitemShut {NoStop}%
\bibitem [{\citenamefont {{Gatti}}(2005)}]{Gatti-2005}%
  \BibitemOpen
  \bibfield  {author} {\bibinfo {author} {\bibfnamefont {C.}~\bibnamefont
  {{Gatti}}},\ }\href {\doibase 10.1524/zkri.220.5.399.65073} {\bibfield
  {journal} {\bibinfo  {journal} {Zeitschrift fur Kristallographie}\ }\textbf
  {\bibinfo {volume} {220}},\ \bibinfo {pages} {399} (\bibinfo {year}
  {2005})}\BibitemShut {NoStop}%
\bibitem [{\citenamefont {Navarrete-L\'opez}\ \emph {et~al.}(2008)\citenamefont
  {Navarrete-L\'opez}, \citenamefont {Garza},\ and\ \citenamefont
  {Vargas}}]{Navarrete-Lopez-2008}%
  \BibitemOpen
  \bibfield  {author} {\bibinfo {author} {\bibfnamefont {A.~M.}\ \bibnamefont
  {Navarrete-L\'opez}}, \bibinfo {author} {\bibfnamefont {J.}~\bibnamefont
  {Garza}}, \ and\ \bibinfo {author} {\bibfnamefont {R.}~\bibnamefont
  {Vargas}},\ }\href {\doibase http://dx.doi.org/10.1063/1.2837666} {\bibfield
  {journal} {\bibinfo  {journal} {The Journal of Chemical Physics}\ }\textbf
  {\bibinfo {volume} {128}},\ \bibinfo {eid} {104110} (\bibinfo {year}
  {2008})}\BibitemShut {NoStop}%
\bibitem [{\citenamefont {Contreras-García}\ and\ \citenamefont
  {Recio}(2011)}]{Contreras-Garcia-2011}%
  \BibitemOpen
  \bibfield  {author} {\bibinfo {author} {\bibfnamefont {J.}~\bibnamefont
  {Contreras-García}}\ and\ \bibinfo {author} {\bibfnamefont {J.}~\bibnamefont
  {Recio}},\ }\href {\doibase 10.1007/s00214-010-0828-1} {\bibfield  {journal}
  {\bibinfo  {journal} {Theoretical Chemistry Accounts}\ }\textbf {\bibinfo
  {volume} {128}},\ \bibinfo {pages} {411} (\bibinfo {year}
  {2011})}\BibitemShut {NoStop}%
\bibitem [{\citenamefont {Rincon}\ \emph {et~al.}(2011)\citenamefont {Rincon},
  \citenamefont {Alvarellos},\ and\ \citenamefont
  {Almeida}}]{Rincon-et-al-2011}%
  \BibitemOpen
  \bibfield  {author} {\bibinfo {author} {\bibfnamefont {L.}~\bibnamefont
  {Rincon}}, \bibinfo {author} {\bibfnamefont {J.~E.}\ \bibnamefont
  {Alvarellos}}, \ and\ \bibinfo {author} {\bibfnamefont {R.}~\bibnamefont
  {Almeida}},\ }\href {\doibase 10.1039/C0CP02711A} {\bibfield  {journal}
  {\bibinfo  {journal} {Phys. Chem. Chem. Phys.}\ }\textbf {\bibinfo {volume}
  {13}},\ \bibinfo {pages} {9498} (\bibinfo {year} {2011})}\BibitemShut
  {NoStop}%
\bibitem [{\citenamefont {de~Silva}\ \emph {et~al.}(0000)\citenamefont
  {de~Silva}, \citenamefont {Korchowiec}, \citenamefont {Ram J.~S.},\ and\
  \citenamefont {Wesolowski}}]{DeSilva-2013}%
  \BibitemOpen
  \bibfield  {author} {\bibinfo {author} {\bibfnamefont {P.}~\bibnamefont
  {de~Silva}}, \bibinfo {author} {\bibfnamefont {J.}~\bibnamefont
  {Korchowiec}}, \bibinfo {author} {\bibfnamefont {N.}~\bibnamefont {Ram
  J.~S.}}, \ and\ \bibinfo {author} {\bibfnamefont {T.~A.}\ \bibnamefont
  {Wesolowski}},\ }\href {\doibase doi:10.2533/chimia.2013.253} {\bibfield
  {journal} {\bibinfo  {journal} {CHIMIA International Journal for Chemistry}\
  }\textbf {\bibinfo {volume} {67}},\ \bibinfo {pages} {253} (\bibinfo {year}
  {2013-04-24T00:00:00})}\BibitemShut {NoStop}%
\bibitem [{\citenamefont {Causà}\ \emph {et~al.}(2013)\citenamefont {Causà},
  \citenamefont {D’Amore}, \citenamefont {Garzillo}, \citenamefont
  {Gentile},\ and\ \citenamefont {Savin}}]{Causa-2013}%
  \BibitemOpen
  \bibfield  {author} {\bibinfo {author} {\bibfnamefont {M.}~\bibnamefont
  {Causà}}, \bibinfo {author} {\bibfnamefont {M.}~\bibnamefont {D’Amore}},
  \bibinfo {author} {\bibfnamefont {C.}~\bibnamefont {Garzillo}}, \bibinfo
  {author} {\bibfnamefont {F.}~\bibnamefont {Gentile}}, \ and\ \bibinfo
  {author} {\bibfnamefont {A.}~\bibnamefont {Savin}},\ }in\ \href {\doibase
  10.1007/978-3-642-32750-6_4} {\emph {\bibinfo {booktitle} {Applications of
  Density Functional Theory to Biological and Bioinorganic Chemistry}}},\
  \bibinfo {series} {Structure and Bonding}, Vol.\ \bibinfo {volume} {150},\
  \bibinfo {editor} {edited by\ \bibinfo {editor} {\bibfnamefont {M.~V.}\
  \bibnamefont {Putz}}\ and\ \bibinfo {editor} {\bibfnamefont {D.~M.~P.}\
  \bibnamefont {Mingos}}}\ (\bibinfo  {publisher} {Springer Berlin
  Heidelberg},\ \bibinfo {year} {2013})\ pp.\ \bibinfo {pages}
  {119--141}\BibitemShut {NoStop}%
\bibitem [{\citenamefont {Karasiev}\ \emph {et~al.}(2006)\citenamefont
  {Karasiev}, \citenamefont {Trickey},\ and\ \citenamefont
  {Harris}}]{Perspectives}%
  \BibitemOpen
  \bibfield  {author} {\bibinfo {author} {\bibfnamefont {V.}~\bibnamefont
  {Karasiev}}, \bibinfo {author} {\bibfnamefont {S.}~\bibnamefont {Trickey}}, \
  and\ \bibinfo {author} {\bibfnamefont {F.}~\bibnamefont {Harris}},\ }\href
  {\doibase 10.1007/s10820-006-9019-8} {\bibfield  {journal} {\bibinfo
  {journal} {Journal of Computer-Aided Materials Design}\ }\textbf {\bibinfo
  {volume} {13}},\ \bibinfo {pages} {111} (\bibinfo {year} {2006})}\BibitemShut
  {NoStop}%
\bibitem [{\citenamefont {Garcia-Aldea}\ and\ \citenamefont
  {Alvarellos}(2007)}]{C3-GEA}%
  \BibitemOpen
  \bibfield  {author} {\bibinfo {author} {\bibfnamefont {D.}~\bibnamefont
  {Garcia-Aldea}}\ and\ \bibinfo {author} {\bibfnamefont {J.~E.}\ \bibnamefont
  {Alvarellos}},\ }\href {\doibase 10.1063/1.2774974} {\bibfield  {journal}
  {\bibinfo  {journal} {The Journal of Chemical Physics}\ }\textbf {\bibinfo
  {volume} {127}},\ \bibinfo {eid} {144109} (\bibinfo {year}
  {2007})}\BibitemShut {NoStop}%
\bibitem [{\citenamefont {Lee}\ \emph {et~al.}(1991)\citenamefont {Lee},
  \citenamefont {Lee},\ and\ \citenamefont {Parr}}]{Lee-Lee-Parr}%
  \BibitemOpen
  \bibfield  {author} {\bibinfo {author} {\bibfnamefont {H.}~\bibnamefont
  {Lee}}, \bibinfo {author} {\bibfnamefont {C.}~\bibnamefont {Lee}}, \ and\
  \bibinfo {author} {\bibfnamefont {R.~G.}\ \bibnamefont {Parr}},\ }\href
  {\doibase 10.1103/PhysRevA.44.768} {\bibfield  {journal} {\bibinfo  {journal}
  {Phys. Rev. A}\ }\textbf {\bibinfo {volume} {44}},\ \bibinfo {pages} {768}
  (\bibinfo {year} {1991})}\BibitemShut {NoStop}%
\bibitem [{\citenamefont {Perdew}\ and\ \citenamefont
  {Constantin}(2007)}]{Lapla1}%
  \BibitemOpen
  \bibfield  {author} {\bibinfo {author} {\bibfnamefont {J.~P.}\ \bibnamefont
  {Perdew}}\ and\ \bibinfo {author} {\bibfnamefont {L.~A.}\ \bibnamefont
  {Constantin}},\ }\href {\doibase 10.1103/PhysRevB.75.155109} {\bibfield
  {journal} {\bibinfo  {journal} {Phys. Rev. B}\ }\textbf {\bibinfo {volume}
  {75}},\ \bibinfo {pages} {155109} (\bibinfo {year} {2007})}\BibitemShut
  {NoStop}%
\bibitem [{\citenamefont {Lee}\ \emph {et~al.}(2009)\citenamefont {Lee},
  \citenamefont {Constantin}, \citenamefont {Perdew},\ and\ \citenamefont
  {Burke}}]{par-z}%
  \BibitemOpen
  \bibfield  {author} {\bibinfo {author} {\bibfnamefont {D.}~\bibnamefont
  {Lee}}, \bibinfo {author} {\bibfnamefont {L.~A.}\ \bibnamefont {Constantin}},
  \bibinfo {author} {\bibfnamefont {J.~P.}\ \bibnamefont {Perdew}}, \ and\
  \bibinfo {author} {\bibfnamefont {K.}~\bibnamefont {Burke}},\ }\href
  {\doibase http://dx.doi.org/10.1063/1.3059783} {\bibfield  {journal}
  {\bibinfo  {journal} {The Journal of Chemical Physics}\ }\textbf {\bibinfo
  {volume} {130}},\ \bibinfo {eid} {034107} (\bibinfo {year}
  {2009})}\BibitemShut {NoStop}%
\bibitem [{\citenamefont {Karasiev}\ \emph {et~al.}(2009)\citenamefont
  {Karasiev}, \citenamefont {Jones}, \citenamefont {Trickey},\ and\
  \citenamefont {Harris}}]{Karasiev-et-al-2009}%
  \BibitemOpen
  \bibfield  {author} {\bibinfo {author} {\bibfnamefont {V.~V.}\ \bibnamefont
  {Karasiev}}, \bibinfo {author} {\bibfnamefont {R.~S.}\ \bibnamefont {Jones}},
  \bibinfo {author} {\bibfnamefont {S.~B.}\ \bibnamefont {Trickey}}, \ and\
  \bibinfo {author} {\bibfnamefont {F.~E.}\ \bibnamefont {Harris}},\ }\href
  {\doibase 10.1103/PhysRevB.80.245120} {\bibfield  {journal} {\bibinfo
  {journal} {Phys. Rev. B}\ }\textbf {\bibinfo {volume} {80}},\ \bibinfo
  {pages} {245120} (\bibinfo {year} {2009})}\BibitemShut {NoStop}%
\bibitem [{\citenamefont {{Trickey}}\ \emph {et~al.}(2009)\citenamefont
  {{Trickey}}, \citenamefont {{Karasiev}},\ and\ \citenamefont
  {{Jones}}}]{3.5}%
  \BibitemOpen
  \bibfield  {author} {\bibinfo {author} {\bibfnamefont {S.~B.}\ \bibnamefont
  {{Trickey}}}, \bibinfo {author} {\bibfnamefont {V.~V.}\ \bibnamefont
  {{Karasiev}}}, \ and\ \bibinfo {author} {\bibfnamefont {R.~S.}\ \bibnamefont
  {{Jones}}},\ }\href {\doibase 10.1002/qua.22312} {\bibfield  {journal}
  {\bibinfo  {journal} {International Journal of Quantum Chemistry}\ }\textbf
  {\bibinfo {volume} {109}},\ \bibinfo {pages} {2943} (\bibinfo {year}
  {2009})}\BibitemShut {NoStop}%
\bibitem [{\citenamefont {Cohen}(1979)}]{C3-LKEIQM}%
  \BibitemOpen
  \bibfield  {author} {\bibinfo {author} {\bibfnamefont {L.}~\bibnamefont
  {Cohen}},\ }\href {\doibase 10.1063/1.437511} {\bibfield  {journal} {\bibinfo
   {journal} {The Journal of Chemical Physics}\ }\textbf {\bibinfo {volume}
  {70}},\ \bibinfo {pages} {788} (\bibinfo {year} {1979})}\BibitemShut
  {NoStop}%
\bibitem [{\citenamefont {Cohen}(1984)}]{C3-RLKE}%
  \BibitemOpen
  \bibfield  {author} {\bibinfo {author} {\bibfnamefont {L.}~\bibnamefont
  {Cohen}},\ }\href {\doibase 10.1063/1.447257} {\bibfield  {journal} {\bibinfo
   {journal} {The Journal of Chemical Physics}\ }\textbf {\bibinfo {volume}
  {80}},\ \bibinfo {pages} {4277} (\bibinfo {year} {1984})}\BibitemShut
  {NoStop}%
\bibitem [{\citenamefont {Anderson}\ \emph {et~al.}(2010)\citenamefont
  {Anderson}, \citenamefont {Ayers},\ and\ \citenamefont
  {Hernandez}}]{C3-HAITLKE}%
  \BibitemOpen
  \bibfield  {author} {\bibinfo {author} {\bibfnamefont {J.~S.~M.}\
  \bibnamefont {Anderson}}, \bibinfo {author} {\bibfnamefont {P.~W.}\
  \bibnamefont {Ayers}}, \ and\ \bibinfo {author} {\bibfnamefont {J.~I.~R.}\
  \bibnamefont {Hernandez}},\ }\href {\doibase 10.1021/jp1029745} {\bibfield
  {journal} {\bibinfo  {journal} {The Journal of Physical Chemistry A}\
  }\textbf {\bibinfo {volume} {114}},\ \bibinfo {pages} {8884} (\bibinfo {year}
  {2010})},\ \Eprint
  {http://arxiv.org/abs/http://pubs.acs.org/doi/pdf/10.1021/jp1029745}
  {http://pubs.acs.org/doi/pdf/10.1021/jp1029745} \BibitemShut {NoStop}%
\bibitem [{\citenamefont {Romanowski}\ and\ \citenamefont
  {Krukowski}(2009)}]{Green}%
  \BibitemOpen
  \bibfield  {author} {\bibinfo {author} {\bibfnamefont {Z.}~\bibnamefont
  {Romanowski}}\ and\ \bibinfo {author} {\bibfnamefont {S.}~\bibnamefont
  {Krukowski}},\ }\href@noop {} {\bibfield  {journal} {\bibinfo  {journal}
  {Acta Physica Polonica Series A}\ }\textbf {\bibinfo {volume} {115}},\
  \bibinfo {pages} {653} (\bibinfo {year} {2009})}\BibitemShut {NoStop}%
\bibitem [{\citenamefont {Karasiev}\ \emph
  {et~al.}(2014{\natexlab{b}})\citenamefont {Karasiev}, \citenamefont
  {Sjostrom},\ and\ \citenamefont {Trickey}}]{ProfessAtQE}%
  \BibitemOpen
  \bibfield  {author} {\bibinfo {author} {\bibfnamefont {V.~V.}\ \bibnamefont
  {Karasiev}}, \bibinfo {author} {\bibfnamefont {T.}~\bibnamefont {Sjostrom}},
  \ and\ \bibinfo {author} {\bibfnamefont {S.}~\bibnamefont {Trickey}},\ }\href
  {\doibase http://dx.doi.org/10.1016/j.cpc.2014.08.023} {\bibfield  {journal}
  {\bibinfo  {journal} {Computer Physics Communications}\ }\textbf {\bibinfo
  {volume} {185}},\ \bibinfo {pages} {3240 } (\bibinfo {year}
  {2014}{\natexlab{b}})}\BibitemShut {NoStop}%
\bibitem [{\citenamefont {{Clementi}}\ and\ \citenamefont
  {{Roetti}}(1974)}]{3.6}%
  \BibitemOpen
  \bibfield  {author} {\bibinfo {author} {\bibfnamefont {E.}~\bibnamefont
  {{Clementi}}}\ and\ \bibinfo {author} {\bibfnamefont {C.}~\bibnamefont
  {{Roetti}}},\ }\href {\doibase 10.1016/S0092-640X(74)80016-1} {\bibfield
  {journal} {\bibinfo  {journal} {Atomic Data and Nuclear Data Tables}\
  }\textbf {\bibinfo {volume} {14}},\ \bibinfo {pages} {177} (\bibinfo {year}
  {1974})}\BibitemShut {NoStop}%
\bibitem [{\citenamefont {Liu}\ \emph {et~al.}(1996)\citenamefont {Liu},
  \citenamefont {Süle}, \citenamefont {L\'opez-Boada},\ and\ \citenamefont
  {Nagy}}]{Liu-et-al-1996}%
  \BibitemOpen
  \bibfield  {author} {\bibinfo {author} {\bibfnamefont {S.}~\bibnamefont
  {Liu}}, \bibinfo {author} {\bibfnamefont {P.}~\bibnamefont {Süle}}, \bibinfo
  {author} {\bibfnamefont {R.}~\bibnamefont {L\'opez-Boada}}, \ and\ \bibinfo
  {author} {\bibfnamefont {A.}~\bibnamefont {Nagy}},\ }\href {\doibase
  http://dx.doi.org/10.1016/0009-2614(96)00515-5} {\bibfield  {journal}
  {\bibinfo  {journal} {Chemical Physics Letters}\ }\textbf {\bibinfo {volume}
  {257}},\ \bibinfo {pages} {68 } (\bibinfo {year} {1996})}\BibitemShut
  {NoStop}%
\end{thebibliography}%

\end{document}